\documentclass[9pt,shortpaper,twoside,web]{IEEEtran}
\usepackage{cite}
\usepackage{amsmath,amssymb,amsfonts}
\usepackage{algorithmic}
\usepackage{graphicx}
\usepackage{subcaption}
\usepackage{textcomp}
\usepackage{color}

\newtheorem{problem}{Problem}
\newtheorem{definition}{Definition}

\newtheorem{theorem}{Theorem}
\newtheorem{lemma}{Lemma}
\newtheorem{corollary}{Corollary}

\newtheorem{remark}{Remark}

\newenvironment{proof}{\begin{trivlist}\item[] {\em Proof: }}{\hfill $\Box$ \end{trivlist}} 

\begin{document}
\title{Transporting Robotic Swarms via Mean-Field Feedback Control}

\author{Tongjia Zheng$^{1}$, Qing Han$^{2}$ and Hai Lin$^{1}$, \IEEEmembership{Senior Member, IEEE}
\thanks{*This work was supported in part by the National Science Foundation under Grant IIS-1724070, and Grant CNS-1830335, and in part by the Army Research Laboratory under Grant W911NF-17-1-0072.}
\thanks{$^{1}$Tongia Zheng and Hai Lin are with the Department of Electrical Engineering, University of Notre Dame, Notre Dame, IN 46556, USA. {\tt\small  tzheng1@nd.edu, hlin1@nd.edu.}} 
\thanks{$^{2}$Qing Han is with the Department of Mathematics, University of Notre Dame, Notre Dame, IN 46556, USA. {\tt\small Qing.Han.7@nd.edu.}} 
}

\maketitle

\begin{abstract}
With the rapid development of AI and robotics, transporting a large swarm of networked robots has foreseeable applications in the near future. Existing research in swarm robotics has mainly followed a bottom-up philosophy with predefined local coordination and control rules. 
However, it is arduous to verify the global requirements and analyze their performance. 
This motivates us to pursue a top-down approach, and develop a provable control strategy for deploying a robotic swarm to achieve a desired global configuration. 
Specifically, we use mean-field partial differential equations (PDEs) to model the swarm and control its mean-field density (i.e., probability density) over a bounded spatial domain using mean-field feedback.
The presented control law uses density estimates as feedback signals and generates corresponding velocity fields that, by acting locally on individual robots, guide their global distribution to a target profile. 
{\color{black} The design of the velocity field is therefore centralized, but the implementation of the controller can be fully distributed -- individual robots sense the velocity field and derive their own velocity control signals accordingly.}
The key contribution lies in applying the concept of input-to-state stability (ISS) to show that the perturbed closed-loop system (a nonlinear and time-varying PDE) is locally ISS with respect to density estimation errors. 
The effectiveness of the proposed control laws is verified using agent-based simulations.
\end{abstract}

\begin{IEEEkeywords}
Input-to-state stability, PDE control systems, Swarm robotics.
\end{IEEEkeywords}

\section{Introduction}
Transporting a large robotic swarm to form certain desired global configuration is a fundamental question for a wide range of applications, such as scheduling transportation systems \cite{teodorovic2008swarm} and employing nanorobots for drug delivery.
Swarm robotic system provides superior robustness and flexibility, but also poses significant challenges in its design \cite{brambilla2013swarm}. We pursue the design of robotic swarms as a control problem, and propose a control theory based design framework. 

The major difficulty of controlling such large-scale systems results from that the control mechanism is expected to be scalable and satisfy the robots' own kinematics while their collective behaviors should be predictable and controllable. 
Existing work has revealed two different philosophies, termed as bottom-up and top-down, respectively \cite{crespi2008top}. 
Probabilistic finite state machines \cite{nouyan2008path} and artificial physics \cite{hettiarachchi2009distributed} are two classic representatives dating back to 1980s which follow the bottom-up philosophy and are known to be decentralized and scalable. However, evaluation of their stability and performance quickly becomes intractable when we increase the swarm size. 
In 2000s, graph theory was introduced into the multi-agent system community and has seen successful applications in designing coordination protocols \cite{cao2012overview}. Nevertheless, it needs to address the dimensionality issue due to large-size matrices arising in swarm robotic systems.
In recent years, top-down design has received increased research interests, which usually employs compact models to describe the macroscopic behaviors. The challenge lies in the appropriate decomposition of the global control strategy into local commands.
Markov chains approach is one representative that uses abstraction-based models for macroscopic descriptions, which partitions the workspace into a grid, over which it defines a probability distribution and designs the movement between cells to govern the evolution of the distribution \cite{accikmecse2012markov, bandyopadhyay2017probabilistic}. The major drawback is that the robots' dynamics are not considered. 
Potential games approach is another top-down approach that adopts a game theoretic formulation, which suggests to decompose the global objective function into local ones that align with the global objective \cite{marden2009cooperative}. However, finding such a decomposition is problem dependent and difficult in general.

{\color{black}
Our work is inspired by the recent top-down design that uses PDEs for macroscopic descriptions. 
There exist mainly two types of PDE models in the literature.
The first type is motivated from the fact that certain discretized PDEs match the dynamics of graph-based coordination algorithms \cite{ferrari2006analysis,meurer2011finite,qi2014multi,pilloni2015consensus,freudenthaler2020pde}.
Boundary control and backstepping design are popular techniques for such models. 
Although boundary actuation is intriguing, it has difficulty for higher-dimensional extension.
Our work adopts the other type that is known as mean-field PDEs \cite{milutinovi2006modeling,lasry2007mean,hamann2008framework,foderaro2014distributed,elamvazhuthi2015optimal,eren2017velocity,krishnan2018distributed,elamvazhuthi2018bilinear}.
These mean-field models fill the gap between individual dynamics and their global behavior with a family of ordinary/stochastic differential equations that describe the motion of individual robots, and a PDE that models the time-evolution of the mean field (i.e. their probability distribution). 
}
A typical design process starts with specifying the task using the macrostate of the PDE, and then computes local motion commands for individuals.
In \cite{elamvazhuthi2015optimal}, the authors formulate an optimization problem for a set of advection-diffusion-reaction PDEs to compute the velocity field and switching rates.
In \cite{foderaro2014distributed}, the authors present a PDE-constrained optimal control problem, and microscopic control laws are derived from the optimal macroscopic description using a potential function approach. 
These optimization-based approaches are however computationally expensive, open-loop, and may be unstable in the presence of unknown disturbance.
Mean field games incorporate the mean-field idea into large population differential games and obtain a compact model with two coupled PDEs \cite{lasry2007mean}.
The control strategy in our work is inspired by the recent idea that uses mean-field feedback to design appropriate velocity fields \cite{eren2017velocity,krishnan2018distributed,elamvazhuthi2018bilinear}. 
Such control laws can be computed efficiently and be formally proved to be convergent. 
Nevertheless, the works \cite{eren2017velocity, krishnan2018distributed} are restricted to deterministic individual dynamics, while stochasticity is ubiquitous in practice, caused possibly by sensor and actuator errors, or the inherent avoidance mechanism of the robots. 
The control law proposed in \cite{elamvazhuthi2018bilinear} applies to the stochastic case, but their focus is on the controllability property. 
{\color{black}
In practice, mean-field feedback control relies on estimating the unknown density, which causes robustness issues in terms of estimation errors. 
Our work distinguishes from \cite{eren2017velocity,krishnan2018distributed,elamvazhuthi2018bilinear} in that we will consider the stochastic case, present general results for its solution property, and study the robustness issue of the proposed control law. 
}

In particular, we study the problem of mean-field feedback control of swarm robotic systems modelled by PDEs. 
We design velocity fields (from which individual control commands can be derived) by using the real-time density as feedback signals, such that the density of the closed-loop system evolves towards a desired one. 
Our contribution includes three aspects. 
{\color{black}
First, we present general results for the solution property (well-posedness, regularity and positivity) of the PDE system.
Second, we propose mean-field feedback laws for robotic swarms that involve stochastic motions and apply the notion of input-to-state stability (ISS) to prove that the closed-loop system is locally ISS with respect to density estimation errors. 
Third, in terms of theoretical contribution to PDE control systems, our results apply the concept of ISS to nonlinear and time-varying PDEs with unbounded operators. Most existing work that studies ISS for PDE systems is restricted to the linear case or one-dimensional case (e.g. boundary control). 
However, the control input in our problem is a vector field that couples with the system state (which thus makes the perturbed closed-loop system nonlinear), and acts on the system through unbounded operators.
}

The rest of the paper is organized as follow. Section \ref{section:preliminaries} introduces some preliminaries and useful lemmas. Problem formulation is given in Section \ref{section:problem formulation}. Section \ref{section:main results} is our main results, in which we present general results for the solution property of the PDE system, present a mean-field feedback control law and then study its robustness issue with respect to density estimation errors. Section \ref{section:simulation} performs an agent-based simulation to verify the effectiveness of the control law. Section \ref{section:conclusion} summarizes the contribution and points out future research.

\section{Preliminaries} \label{section:preliminaries}
\subsection{Notations and useful lemmas}\label{section:notation}
{\color{black}
Let $E\subset\mathbb{R}^n$ be a measurable set and $k\in\mathbb{N}$. Consider $f:E\to\mathbb{R}$. Denote $C^k(E)=\{f|f^{(k)}\text{ is continuous}\}$ and $C(E)=C^0(E)$.
For $p\in[1,\infty)$, denote $L^{p}(E)=\{f|\|f\|_{L^p(\omega)}:=(\int_{E}|f(x)|^{p}dx)^{1/p}<\infty\}$, endowed with the norm $\|\cdot\|_{L^p(E)}$. Denote $L^{\infty}(E)=\{f|\|f\|_{L^\infty(E)}:=\operatorname{ess}\sup_{x\in E}|f(x)|<\infty\}$, endowed with the norm $\|f\|_{L^\infty(E)}$. 
We use $D^\alpha f$ to represent the weak derivatives of $f$ for all multi-indices $\alpha$ of length $|\alpha|$.
For $p\in[1,\infty)$, denote $W^{k, p}(E)=\{f|\|f\|_{W^{k, p}(E)}:=\sum_{|\alpha|\leq k} \|D^{\alpha}f\|_{L^p(E)}<\infty\}$, endowed with the norm $\|\cdot\|_{W^{k, p}(E)}$.
Analogously, $W^{k,\infty}(E)$ is defined, equipped with the norm $\|\cdot\|_{W^{k,\infty}(E)}$. 
We also denote $H^k=W^{k,2}$.
The spaces $W^{k, p}(E)$ are referred to as \textit{Sobolev spaces}.

Let $\omega$ be a bounded and connected $C^1$-domain in $\mathbb R^n$ and $T>0$ be a constant. Denote by $\partial\omega$ the boundary of $\omega$.
Set $\Omega=\omega\times(0,T)$. 
For a function $u(x,t):\Omega\to\mathbb{R}$, we call $x$ the spatial variable and $t$ the time variable. We denote $\partial_tu=\partial u/\partial t$ and $\partial_iu=\partial u/\partial x_i$, where $x_i$ is the $i$-th coordinate of $x$. 
The gradient and Laplacian of a scalar function $f$ are denoted by $\nabla f$ and $\Delta f$, respectively, and the divergence of a vector field $\boldsymbol{F}$ is denoted by $\nabla\cdot\boldsymbol{F}$. 
The differentiation operation of these operators are only taken with respect to the spatial variable $x$ if $f$ and $\boldsymbol{F}$ are also functions of $t$.

{\color{black}
We define the following space of time and space dependent functions, which will be used to study the solution of PDEs:
\begin{align*}
\begin{split}
    \mathcal{M}:=\{&u\in L^{2}(\Omega)| \partial_{i}u\in L^{2}(\Omega), i=1,\dots,n\\ &\text{and } u(\cdot,t)\in L^{2}(\omega),\forall
t\in[0,T]\}.
\end{split}
\end{align*}}
The norm on $\mathcal{M}$ is defined by 
$$
\|u\|_\mathcal{M}=\sum_{i=1}^n\|\partial_iu\|_{L^2(\Omega)}+\sup_{t\in [0,T]}\|u(\cdot, t)\|_{\color{black}{L^2(\omega)}}.
$$ 

Throughout this paper, when we say $f(x)$ is a \textit{density function} on $\omega$, we mean $f(x)$ is a probability density function, i.e. $f(x)\geq0,\forall x\in\omega$ and $\int_\omega f(x)dx=1$.
}


\begin{lemma} \label{lmm:Poincare inequality}
(Poincar\'e inequality \cite{lieberman1996second}). For $p \in[1, \infty)$ and $\omega$, a bounded connected open set of $\mathbb{R}^{n}$ with a Lipschitz boundary, there exists a constant $C$ depending only on $\omega$ and $p$ such that for every function $f\in W^{1, p}(\omega)$,
$$
\left\|f-f_{\omega}\right\|_{p} \leq C\|\nabla f\|_{p},
$$
where $f_{\omega}=\frac{1}{|\omega|} \int_{\omega} f d x$, and $|\omega|$ is the Lebesgue measure of $\omega$.
\end{lemma}

\subsection{Input-to-state stability}
{\color{black}
Input-to-state stability is a stability notion widely used to study stability of nonlinear control systems with external inputs \cite{sontag1995characterizations}. 
We introduce its extension to infinite-dimensional systems presented in \cite{dashkovskiy2013input}.
Let $\left(X,\|\cdot\|_{X}\right)$ and $\left(U,\|\cdot\|_{U}\right)$ be the state space and the space of input values, endowed with norms $\|\cdot\|_{X}$ and $\|\cdot\|_{U}$, respectively.
Denote by $PC(I; Y)$ the space of piecewise right-continuous functions from $I\subset\mathbb{R}$ to $Y$, equipped with the standard sup-norm.
Define the following classes of comparison functions:
\begin{align*}
    \mathcal{K} &:= \{\gamma :\mathbb{R}_{+} \to \mathbb{R}_{+}|\gamma \text{ is continuous and strictly} \\
    &\quad\quad \text{increasing with } \gamma (0)=0\}\\
    \mathcal{K}_{\infty} &:=\{\gamma \in \mathcal{K}|\gamma \text{ is unbounded}\}  \\
    \mathcal{L} &:= \{\gamma  : \mathbb{R}_{+} \to \mathbb{R}_{+} | \gamma \text{ is continuous and strictly } \\
    &\quad\quad \text{decreasing with } \lim_{t\to \infty} \gamma (t) = 0\}  \\
    \mathcal{KL} &:=\{\beta:\mathbb{R}_{+} \times\mathbb{R}_{+} \to\mathbb{R}_{+} |\beta(\cdot, t) \in \mathcal{K}, \forall t \geq 0, \\
    &\qquad\beta(r, \cdot) \in \mathcal{L}, \forall r>0\}.
\end{align*}

We use the following axiomatic definition of a control system \cite{dashkovskiy2013input}.
\begin{definition}\label{dfn:control system}
The triple $\Sigma=\left(X, U_{c}, \phi\right)$, consisting of the state space $X$, the space of admissible input functions $U_{c} \subset\left\{f: \mathbb{R}_{+} \rightarrow U\right\}$, both of which are linear normed spaces, equipped with norms $\|\cdot\|_{X}$ and $\|\cdot\|_{U_{c}}$, respectively, and of a transition map $\phi: A_{\phi} \rightarrow X, A_{\phi} \subset \mathbb{R}_{+} \times \mathbb{R}_{+} \times X \times U_{c}$, is called a control system, if the following
properties hold:
\begin{itemize}
    \item Existence: for every $\left(t_{0}, \phi_{0}, u\right) \in \mathbb{R}_{+} \times X \times U_{c}$ there exists $t>t_{0}:\left[t_{0}, t\right] \times\left\{\left(t_{0}, \phi_{0}, u\right)\right\} \subset A_{\phi}$.
    \item Identity property: for every $\left(t_{0}, \phi_{0}\right) \in \mathbb{R}_{+} \times X$ it holds $\phi\left(t_{0}, t_{0}, \phi_{0}, \cdot\right)=\phi_{0}$.
    \item Causality: for every $\left(t, t_{0}, \phi_{0}, u\right) \in A_{\phi}$, for every $\tilde{u} \in U_{c}$, such that $u(s)=$ $\tilde{u}(s), s \in\left[t_{0}, t\right]$ it holds $\left(t, t_{0}, \phi_{0}, \tilde{u}\right) \in A_{\phi}$ and $\phi\left(t, t_{0}, \phi_{0}, u\right) \equiv \phi\left(t, t_{0}, \phi_{0}, \tilde{u}\right)$.
    \item Continuity: for each $\left(t_{0}, \phi_{0}, u\right) \in \mathbb{R}_{+} \times X \times U_{c}$ the map $t \mapsto \phi\left(t, t_{0}, \phi_{0}, u\right)$ is continuous.
    \item Semigroup property: for all $t \geq s \geq 0$, for all $\phi_{0} \in X, u \in U_{c}$ so that $\left(t, s, \phi_{0}, u\right) \in$ $A_{\phi}$, it follows
    \begin{itemize}
        \item $\left(r, s, \phi_{0}, u\right) \in A_{\phi}, r \in[s, t]$.
        \item for all $r \in[s, t]$ it holds $\phi(t, r, \phi(r, s, x, u), u)=\phi(t, s, x, u)$.
    \end{itemize}
\end{itemize}
\end{definition} 

Here, $\phi(t, s, x, u)$ denotes the system state at time $t \in \mathbb{R}_{+}$, if its state at time $s \in \mathbb{R}_{+}$ was $x \in X$ and the input $u \in U_{c}$ was applied.

\begin{definition} \label{dfn:(L)ISS}
$\Sigma$ is called \textit{locally input-to-state stable (LISS)}, if $\exists \rho_{x}, \rho_{u}>0$ and $\exists \beta \in \mathcal{K} \mathcal{L}$ and $\gamma \in \mathcal{K}$, such that the inequality
\begin{equation}\label{eq:(L)ISS}
    \left\|\phi\left(t,t_0,\phi_{0}, u\right)\right\|_{X} \leq \beta\left(\left\|\phi_{0}\right\|_{X}, t-t_0\right)+\gamma\left(\|u\|_{U_{c}}\right)
\end{equation}
holds $\forall \phi_{0}:\left\|\phi_{0}\right\|_{X} \leq \rho_{x},\forall u \in U_{c}:\|u\|_{U_{c}} \leq \rho_{u}$ and $\forall t \geq t_0$.
\end{definition} 

The control system is called \textit{input-to-state stable (ISS)}, if in the above definition $\rho_{x}$ and $\rho_{u}$ can be chosen equal to $\infty$.
If $U_{c}=PC\left(\mathbb{R}_{+}; U\right)$, then $\|u\|_{U_{c}}=\sup_{0\leq s\leq \infty}\|u(s)\|_{U}$, and due to the causality property of $\Sigma$, one can obtain an equivalent definition of (L)ISS by replacing \eqref{eq:(L)ISS} with the following inequality \cite{dashkovskiy2013input}:
\begin{equation*}
    \left\|\phi\left(t,t_0,\phi_{0}, u\right)\right\|_{X} \leq \beta\left(\left\|\phi_{0}\right\|_{X}, t-t_0\right)+\gamma\Big(\sup _{0 \leq s \leq t}\|u(s)\|_{U}\Big).
\end{equation*}
To verify the ISS property, Lyapunov functions can be exploited. 

\begin{definition}\label{dfn:(L)ISS-Lyapunov function}
A continuous function $V:\mathbb{R}_{+}\times D \rightarrow \mathbb{R}_{+}, D \subset X, 0 \in\operatorname{int}(D)=D \backslash \partial D$ is called an \textit{LISS-Lyapunov function} for $\Sigma$, if there exist constants $\rho_{x}, \rho_{u}>0$, functions $\psi_{1}, \psi_{2} \in \mathcal{K}_{\infty}, \chi \in \mathcal{K}$, and a continuous positive definite function $W$, such that:
\begin{itemize}
    \item[(i)] $\psi_{1}\left(\|x\|_{X}\right) \leq V(t, x) \leq \psi_{2}\left(\|x\|_{X}\right), \quad \forall t \in \mathbb{R}_{+}, \forall x \in D$
    \item[(ii)] $\forall x \in X:\|x\|_{X} \leq \rho_{x}, \forall u \in U_{c}:\|u\|_{U_{c}} \leq \rho_{u}$ it holds:
    \begin{equation}\label{eq:LISS implication form 1}
        \|x\|_{X} \geq \chi(\|u\|_{U_{c}}) \Rightarrow \dot{V}_{u}(t,x) \leq-W(\|x\|_{X}), \quad \forall t\in\mathbb{R}_{+}
    \end{equation}
\end{itemize}
where the derivative of $V$ corresponding to the input $u$ is given by
\[
\dot{V}_{u}(t,x)=\varlimsup_{\Delta t \rightarrow+0} \frac{1}{\Delta t}(V(t+\Delta t,\phi(t+\Delta t,t,x,u))-V(t,x)).
\]
\end{definition} 

If in the previous definition $D=X, \rho_{x}=\infty$ and $\rho_{u}=\infty,$ then the function $V$ is called \textit{ISS-Lyapunov function}.
If $U_{c}=PC\left(\mathbb{R}_{+};U\right)$, then condition (ii) in Definition \ref{dfn:(L)ISS-Lyapunov function} can be replaced by the following condition (ii)$'$ due to the causality property \cite{dashkovskiy2013input}:
\begin{itemize}
    \item[(ii)$'$] $\forall x \in X:\|x\|_{X} \leq \rho_{x}, \forall \xi \in U:\|\xi\|_{U} \leq \rho_{u}$ it holds:
    \[
    \|x\|_{X} \geq \chi(\|\xi\|_{U}) \Rightarrow \dot{V}_{u}(t,x) \leq-W(\|x\|_{X}), \quad \forall t\in\mathbb{R}_{+}
    \]
    for all $u \in U_{c}:\|u\|_{U_{c}} \leq \rho_{u}$ with $u(0)=\xi$.
\end{itemize}
}

\section{Problem formulation} \label{section:problem formulation}
This work studies the transport problem of robotic swarms. 
Specifically, we want to design velocity commands for individual robots such that the swarm evolves to certain global distribution. 
The robots are assumed homogeneous whose motions satisfy:
\begin{equation} \label{eq:Langevin equation}
    dX_i={v} (X_i,t)dt+\sqrt{2\sigma(X_i,t)}dB_t, \quad i = 1,\dots,N,
\end{equation}
where $N$ is robots' population, $X_i\in\omega$ is the position of the $i$-th robot, ${v}(X_i,t)=(v_1,\dots,v_n)\in\mathbb{R}^n$ is the velocity field that acts on the robots, $B_t\in\mathbb{R}^n$ is an $n$-dimensional Wiener process which represents stochastic motions, and $\sqrt{2\sigma(X_i,t)}\in\mathbb{R}$ is the standard deviation of the stochastic motion at position $X_i$.

Their macroscopic state can be described by the following mean-field PDE, also known as the Fokker-Planck equation, which models the evolution of the swarm's mean-field density $p(x,t)$ on $\omega$:
\begin{align} \label{eq:FP equation}
\begin{split}
     \partial_t p =-\nabla\cdot({v} p) + \Delta(\sigma p) &\quad\text{in}\quad \Omega, \\
    p=p_0 &\quad\text{on}\quad {\color{black}\omega\times\{0\}},\\
    \boldsymbol{n} \cdot(\nabla(\sigma p)-{v}p)=0 &\quad\text{on}\quad S(\Omega),
\end{split}
\end{align}
where $\boldsymbol{n}$ is the unit inner normal to the boundary $\partial\omega$, and $p_0(x)$ is the initial density. The last equation is the \textit{reflecting boundary condition} to confine the swarm within the domain $\omega$. 

\begin{remark}
We point out that \eqref{eq:FP equation} holds regardless of the number of robots. 
However, if $N$ is small, using the swarm's (probability) density to represent its global state doesn't make much sense. 
Hence, we usually assume $N$ is large. Note that \eqref{eq:Langevin equation} and \eqref{eq:FP equation} share the same set of coefficients, which means that the macroscopic velocity field we design for the PDE system can be easily transmitted to individual robots. 
{\color{black}
Note that individual robots need to derive their own low-level controller to track the reference velocity command, which can however be done in a distributed way.
The velocity tracking problem has been widely studied in literature especially for mobile robots, and hence is not studied in this paper.
}
\end{remark}

\begin{problem}
Given a desired density $p^*(x)$, we want to design the velocity field ${v}(x,t)$ such that the solution of \eqref{eq:FP equation} converges to $p^*(x)$.
\end{problem}

\section{Main results} \label{section:main results}

\subsection{Well-posedness and regularity}
{\color{black}
Before presenting the velocity laws, we shall study the solution property of \eqref{eq:FP equation}, including its well-posedness, regularity, mass conservation and positivity preservation.
}
We point out that \eqref{eq:FP equation} is a special case of the so-called \textit{conormal derivative problem} for parabolic equations of divergent form \cite{lieberman1996second}. 
Relevant results from Chapter VI in \cite{lieberman1996second} are summarized in Appendices.
Considering the initial/boundary value problem \eqref{eq:FP equation}, we have the following result for its weak solutions (see Definition \ref{dfn:weak solution} in Appendices).

\begin{theorem}\label{thm:well-posedness}
Assume 
\begin{equation}\label{eq:regularity condition1}
    v_i\in L^\infty(\Omega), \sigma\in L^\infty(\Omega), \partial_i\sigma\in L^\infty(\Omega)\text{ and } p_0\in L^\infty(\omega).
\end{equation} 
Then we have the following properties:
\begin{itemize}
    \item (\textbf{Well-posedness and regularity}) There exists a unique weak solution $p\in \mathcal{M}$ of the problem \eqref{eq:FP equation}.
    \item (\textbf{Mass conservation}) The solution satisfies $p(\cdot,t)\in{\color{black} H^1(\omega)}$ and $\int_{{\color{black}\omega}}p(\cdot,t)dx=1$ for almost every $t\in(0,T]$.
    \item (\textbf{Positivity preservation}) If we further assume that
    \begin{equation}\label{eq:regularity condition2}
        \partial_iv_i\in L^\infty(\Omega)\text{ and } \partial_i^2\sigma\in L^\infty(\Omega),
    \end{equation} 
    then $p_0\geq(\text{or}>)0$ implies $p\geq(\text{or}>)0$ for almost every $t\in[0,T]$.
\end{itemize}
\end{theorem}

\begin{proof}
We rewrite \eqref{eq:FP equation} as
\begin{align*}
\begin{split}
    Lp=-\partial_tp+\nabla\cdot[\sigma\nabla p+(\nabla\sigma-v)p]=0 &\quad\text{in}\quad \Omega,\\
    p=p_0 &\quad\text{on}\quad {\color{black}\omega\times\{0\}},\\
    Mp=\boldsymbol{n}\cdot(\sigma\nabla p+(\nabla\sigma-v)p)=0 &\quad\text{on}\quad S(\Omega).
\end{split}
\end{align*}
Comparing it with the standard conormal derivative problem \eqref{eq:conormal problem} in Appendices, we note that $a_{ij}=\sigma\delta_{ij}\in L^\infty(\Omega)$, $b_i=\partial_i\sigma-v_i\in L^\infty(\Omega)$, $\varphi=p_0\in L^\infty(\omega)\subset L^2(\omega)$, and all other coefficients in \eqref{eq:conormal problem} are 0. According to Theorem \ref{thm:well-posedness of conormal}, there exists a unique weak solution $p\in \mathcal{M}$ of the problem \eqref{eq:FP equation}.
For such a weak solution, by taking the test function $\eta=1$ in \eqref{eq:weak solution}, we have, for almost every $t\in[0,T]$, ${\color{black}\int_{{\color{black}\omega}}p(x,t)dx=\int_{\omega}p_0(x)dx=1,}$
which means the solution always represents a density function. Also, $p\in\mathcal{M}$ implies that {\color{black}$p(\cdot,t)\in H^1(\omega)$} for almost every $t\in(0,T]$. 
Furthermore, condition \eqref{eq:regularity condition2} implies $\partial_ib_i\in L^\infty(\Omega)$. By Corollary \ref{corollary:positivity}, $p_0\geq(\text{or}>)0$ implies $p\geq(\text{or}>)0$ for almost every $t\in[0,T]$.
\end{proof}

\begin{remark}
The regularity conditions for $\sigma$ and $p_0$ in \eqref{eq:regularity condition1} and \eqref{eq:regularity condition2} can be easily satisfied, while the regularity condition for $v$ depends on the velocity field we design. 
We shall further study this problem in subsequent sections. 
We point out that such a weak solution has $L^2$ spatial derivatives.
{\color{black}
It can be derived from Definition \ref{dfn:weak solution} and \eqref{eq:weak solution} that such a weak solution satisfies $\partial_tp\in L^2(\Omega)$, which implies that $p(t,\cdot)$ is absolutely continuous from $(0,T)$ to $L^2(\omega)$.
This time regularity will enable us to use Lyapunov functions to study its stability.
}
\end{remark}

\subsection{Exponentially stable mean-field feedback control}
First, we present a mean-field feedback law with exponential convergence assuming $p(x,t)$ is available. 
Given a desired density $p^*(x)>0$, define $\Phi(x,t) = p(x,t) - p^*(x)$. 
{\color{black}Denote $\Phi_0=p_0-p^*$.}

{\color{black}
Our main idea is to design ${v}(x,t)$ such that $\Phi(x,t)$ satisfies the following diffusion equation:
\begin{equation} \label{eq:diffusion equation}
    \partial_t\Phi(x,t) = \nabla\cdot[\alpha(x,t)\nabla\Phi(x,t)],
\end{equation}
where $\alpha(x,t)>0$ is the diffusion coefficient. It is known that under mild conditions on $\alpha(x,t)$, the solution of \eqref{eq:diffusion equation} evolves towards a constant function in $\omega$, which will be 0 because $\Phi$ is the difference of two density functions and, for any $t$,
{\color{black}
\begin{equation}
    \int_{\omega}\Phi(x,t)dx=\int_{\omega} p(x,t)dx-\int_{\omega} p^*(x)dx\equiv1-1=0.
\end{equation}
}
The idea of using diffusion/heat equations for designing velocity fields is originally from \cite{eren2017velocity}. 
{\color{black}
Our work extends the original work in three aspects.
First, we generalize the design to PDEs that contains stochastic motions and rigorously study its solution property to justify the Lyapunov-based stability analysis. 
Second, the control law given in \cite{eren2017velocity} can be problematic if the density becomes zero.
We will show how to avoid this issue by appropriately constructing the density estimate in the mean-field feedback law.
Third, we continue to study the robustness of this modified feedback law with respect to density estimation errors (which includes not only the inherent error of any estimation algorithm, but also the ``artificial error'' introduced to ensure that the feedback law remains bounded).

Our first result is to enhance the stability result in \cite{eren2017velocity} assuming that the density is strictly positive and can be perfectly measured.
}
\begin{theorem} \label{thm:exponential stability}
(\textbf{Exponential stability}). Design the velocity field as
\begin{equation} \label{eq:density feedback law}
    {v}(x,t) = -\frac{\alpha(x,t)\nabla\big[p(x,t) - p^*(x)\big] -\nabla\big[\sigma(x,t)p(x,t)\big]}{p(x,t)},
\end{equation}
where $\alpha(x,t)>0$ is a parameter that satisfies $\sup_{(x,t)\in\Omega}\alpha(x,t)<\infty$ and $\inf_{(x,t)\in\Omega}\alpha(x,t)>0$. If the solution satisfies $p(x,t)>0$ for all $t>0$, then $\|\Phi\|_{L^2(\omega)}\to0$ exponentially.
\end{theorem}
\begin{proof}
Substituting \eqref{eq:density feedback law} into \eqref{eq:FP equation}, we obtain the closed-loop PDE
\begin{align*} 
\begin{split}
    \partial_t\Phi = \nabla\cdot(\alpha\nabla\Phi) & \quad\text{in}\quad \Omega, \\
    \Phi=\Phi_0 & \quad\text{on}\quad {\color{black}\omega\times\{0\}},\\
    \boldsymbol{n} \cdot\nabla\Phi=0 & \quad\text{on}\quad S(\Omega),
\end{split}
\end{align*}
which is a diffusion equation with Neumann boundary condition. 
Its exponential convergence is well-known \cite{evans1998partial}. 
We include the proof for completeness.
Consider a Lyapunov function $V(t)=\frac{1}{2}\|\Phi\|_{L^2({\color{black}\omega})}^2=\frac{1}{2}\int_{{\color{black}\omega}}\Phi^2 dx$.
Define $\alpha_{\text{min}}(t):=\inf_{x\in{\color{black}\omega}}\alpha(x,t)>0$. We have
\begin{align*}
    \Dot{V}(t) &= \int_{{\color{black}\omega}}\Phi \partial_t\Phi dx
    = \int_{{\color{black}\omega}}\Phi \nabla\cdot[\alpha(x,t)\nabla\Phi]dx\\
    &= \int_{\partial{{\color{black}\omega}}}\Phi[\alpha(x,t)\nabla\Phi\cdot\boldsymbol{n}]ds -\int_{{\color{black}\omega}} \alpha(x,t)\nabla\Phi\cdot\nabla\Phi dx \\
    &= - \alpha_{\text{min}}(t)\int_{{\color{black}\omega}} |\nabla\Phi|^2 dx 
    \leq -\frac{\alpha_{\text{min}}(t)}{C^2} \int_{\omega} |\Phi|^2 dx,
\end{align*}
{\color{black} where we used divergence theorem for the third equality, the boundary condition for the forth equality, and Poincar\'e inequality (for which we also use the fact that $\int_\omega\Phi dx\equiv0$) for the inequality, and $C>0$ is a constant depending on $\omega$. 
}
Since $\alpha$ has a uniform positive lower bound, we obtain exponential stability.
\end{proof}



The control law \eqref{eq:density feedback law} generates a dynamic velocity field on $\omega$ based on the real-time density {\color{black}(in a centralized way)}, where $\alpha$ is a design parameter for adjusting the local velocity magnitude. 
By following this velocity command, the swarm density is guaranteed to evolve towards $p^*(x)$.
In implementation, each robot computes its velocity command ${v}(X_i,t)$ using only function values around its position $X_i$ and then derives its own velocity tracking controller {\color{black}in a distributed way}. 
Thus, this control strategy is computationally efficient and scalable to swarm sizes. 
Note that for \eqref{eq:density feedback law} to be well-defined, we require its denominator $p>0$. This problem will be addressed when we replace $p$ with an estimated density later.
}

\subsection{Mean-field feedback control using density estimates}
The exponential stability result requires that $p(x,t)$ is known and positive for all $x\in\omega$. 
In this section, we use kernel density estimation (KDE) to obtain an estimate of $p$ that is always positive, and study its robustness with respect to estimation errors.

KDE is a non-parametric way to estimate an unknown density \cite{silverman1986density}. The robots' positions $\{X_{i}(t)\}_{i=1}^N$ can be seen as a set of $N$ samples of the common density $p(x,t)$. 
The density estimator is given by
\begin{equation} \label{eq:KDE}
\hat{p}(x,t) = \frac{1}{N h^{n}} \sum_{i=1}^{N} K\left(\frac{1}{h}\left(x-X_{i}(t)\right)\right),
\end{equation}
where $K(x)$ is a kernel function chosen to be the Gaussian kernel
$$
K(x)=\frac{1}{(2 \pi)^{n/2}} \exp \left(-\frac{1}{2}x^\intercal x\right),
$$
and $h$ is the bandwidth, usually chosen as a function of $N$ such that $\lim _{N \rightarrow \infty} h_N=0$ and $\lim _{N \rightarrow \infty} N h_{N}=\infty$. (Many boundary correction methods exist for refining the density estimate to have compact support \cite{silverman1986density}, so we shall not worry about this issue.) These estimates (and their derivatives) are in general uniformly consistent in the sense that $\lim_{N \rightarrow\infty}\|\hat{p}_{N}-p\|_{L^\infty}=0$ with probability 1. 

{\color{black}
With the density estimate, the control law is changed to
\begin{equation} \label{eq:density feedback law using estimation}
    {v}(x,t) = -\frac{\alpha(x,t)\nabla\big[\Hat{p}(x,t) - p^*(x)\big] -\nabla \big[\sigma(x,t)\Hat{p}(x,t)\big]}{\Hat{p}(x,t)},
\end{equation}
which is well-defined since $\Hat{p}>0$ with our choice of kernels.
\begin{remark}
Since the Gaussian kernels are positive $C^\infty$ functions and the Wiener processes have continuous paths, we have that $\Hat{p}\in C^\infty(\omega)\times C([0,T])$. 
Moreover, since $\omega$ is bounded and $N$ is finite, we have that $\inf_{(x,t)\in\Omega}\Hat{p}(x,t)>0$ and $\sup_{(x,t)\in\Omega}\partial_i^k\Hat{p}(x,t)<\infty$ for any $k\in\mathbb{N}$ as long as we fix $h$.
Hence, if $\alpha\in W^{1,\infty}(\omega)\times L^\infty([0,T])$, $p^*\in W^{2,\infty}(\omega)$ and $\sigma\in W^{2,\infty}(\omega)\times L^\infty([0,T])$, then $v_i\in L^\infty(\Omega)$ with $\partial_iv_i\in L^\infty(\Omega)$, i.e. $v$ satisfies the regularity conditions \eqref{eq:regularity condition1} and \eqref{eq:regularity condition2} in Theorem \ref{thm:well-posedness}.
\end{remark}

Now we study the robustness issue in terms of density estimation errors. 
{\color{black}Such errors can arise from not only the inherent error of any estimation algorithm, but also some ``artificial error'' we impose on $\Hat{p}$ to ensure that the feedback law \eqref{eq:density feedback law using estimation} remains bounded.
}
Since $p>0$, we can define $\epsilon(x,t):=\Hat{p}(x,t)/p(x,t)-1$, or equivalently $\Hat{p} = p\big(1 + \epsilon)$. Then $\epsilon=0$ if and only if $\Hat{p}=p$, for which we view $\epsilon$ as estimation errors. We also have $\epsilon>-1$ since $\Hat{p}>0$. Our idea is to treat a functional of $\epsilon(x,t)$, denoted by $d(t)$ (defined later), as external input and establish ISS property with respect to $d(t)$. In this way, the perturbed closed-loop system will be bounded by a function of $d(t)$ and be asymptotically stable when $d(t)=0$. 

{\color{black}
First, substituting $p=\Phi+p^*$ into \eqref{eq:FP equation}, then $\Phi$ satisfies
\begin{align} \label{eq:system of Phi}
\begin{split}
    \partial_t\Phi=-\nabla\cdot[{v}(\Phi+p^*)] + \Delta[\sigma(\Phi+p^*)] &\quad\text{in}\quad \Omega, \\
    \Phi=\Phi_0 &\quad\text{on}\quad {\color{black}\omega\times\{0\}},\\
    \boldsymbol{n}\cdot\big[{v}(\Phi+p^*) + \nabla\big(\sigma(\Phi+p^*)\big)\big]=0 &\quad\text{on}\quad S(\Omega). 
\end{split}
\end{align}
Since $p$ is a weak solution of \eqref{eq:FP equation}, then $\Phi$ is also a weak solution of \eqref{eq:system of Phi}.
Now, substitute \eqref{eq:density feedback law using estimation} into \eqref{eq:system of Phi}, and use $\Hat{p}=p(1+\epsilon)$. 
We obtain
\begin{align}
\begin{split}
    \partial_t\Phi &= \nabla\cdot\Big( \frac{p[\alpha\nabla(\Hat{p}-p^*)-\nabla(\sigma\Hat{p})]}{\Hat{p}} + \nabla(\sigma p) \Big)\\
    &= \nabla\cdot\frac{\alpha \nabla [\Phi(1+\epsilon)] + \alpha \nabla(\epsilon p^*) - \sigma(\Phi+p^*) \nabla\epsilon}{1+\epsilon} \\
    &= \nabla\cdot(\alpha\nabla\Phi) + \nabla\cdot\Big((\alpha-\sigma)\Phi\frac{\nabla\epsilon}{1+\epsilon}\Big) \\
    &\quad + \nabla\cdot\Big((\alpha-\sigma)p^*\frac{\nabla\epsilon}{1+\epsilon}\Big) + \nabla\cdot\Big(\alpha\nabla p^*\frac{\epsilon}{1+\epsilon}\Big).
\end{split}
\end{align}
Define $u_1=\frac{\nabla\epsilon}{1+\epsilon}$ and $u_2=\frac{\epsilon}{1+\epsilon}$.
Then the perturbed closed-loop system is given by
\begin{align}\label{eq:perturbed system}
\begin{split}
    \partial_t\Phi &= \nabla\cdot(\alpha\nabla\Phi) + \nabla\cdot\big((\alpha-\sigma)\Phi u_1\big) \\
    & \quad + \nabla\cdot\big((\alpha-\sigma)p^*u_1\big) + \nabla\cdot\big(\alpha u_2\nabla p^*\big)  \quad\text{in}\quad \Omega,
\end{split}
\end{align}
with initial and boundary conditions
\begin{align*}
\begin{split}
    &\Phi=\Phi_0 \text{ on } {\color{black}\omega\times\{0\}},\\
    &\boldsymbol{n}\cdot\big[\alpha\nabla\Phi + (\alpha-\sigma)\Phi u_1 + (\alpha-\sigma)p^*u_1 + \alpha u_2\nabla p^*\big]=0 \text{ on } S(\Omega).
\end{split}
\end{align*}
By defining
\begin{align*}
\begin{split}
    &A_1f=\nabla\cdot(\alpha\nabla f), \qquad\qquad A_2(f,g)=\nabla\cdot\big((\alpha-\sigma)fg\big), \\
    &B_1f=\nabla\cdot\big((\alpha-\sigma)p^*f\big), \quad B_2f=\nabla\cdot\big(\alpha f\nabla p^*\big),
\end{split}
\end{align*}
we can rewrite \eqref{eq:perturbed system} in a form of an abstract bilinear control system:
\begin{align}\label{eq:abstract bilinear system}
\begin{split}
    \Dot{\Phi}=A_1\Phi + A_2(\Phi,u_1) + B_1u_1 + B_2u_2,\quad \Phi(0)=\Phi_0
\end{split}
\end{align}
where $A_1,B_1,B_2$ are linear operators and $A_2$ is bilinear.
Hence, this system is essentially nonlinear.
To study its ISS property, we first present the following theorem which exploits Lyapunov functions to verify the ISS property for nonlinear and time-varying infinite-dimensional control systems.

\begin{theorem}\label{thm:(L)ISS-Lyapunov function}
Let $\Sigma=\left(X, U_{c}, \phi\right)$ be a control system, and $x \equiv 0$ be its equilibrium point.
Assume for all $u \in U_{c}$ and for all $s \geq 0$ a function $\tilde{u}$, defined by $\tilde{u}(\tau)=u(\tau+s)$ for all $\tau \geq 0$, belongs to $U_{c}$ and $\|\tilde{u}\|_{U_{c}} \leq\|u\|_{U_{c}}$.
If $\Sigma$ possesses an (L)ISS-Lyapunov function, then it is (L)ISS.
\end{theorem}

\begin{proof}
The proof is included in the Appendices.
It is based on the proof in \cite{dashkovskiy2013input}, but extends it to time-varying control systems.
\end{proof}

The assumption of $U_c$ in Theorem \ref{thm:(L)ISS-Lyapunov function} holds for many usual function classes, including $PC\left(\mathbb{R}_{+}; U\right), L^{p}\left(\mathbb{R}_{+}; U\right), p\in[1,\infty]$,
Sobolev spaces, etc \cite{dashkovskiy2013input}.
In our problem, $u_1,u_2\in PC\left(\mathbb{R}_{+}; U\right)$ because Wiener processes have continuous paths.

\begin{remark}
We point out that in the development of the ISS notion, there is no reference to specific notion of solution.
Instead, it is based on the concept of an abstract control system defined in Definition \ref{dfn:control system}.
Hence, as long as the notion of solution (e.g. weak, mild, strong, classical) of the infinite-dimensional problem is selected such that the properties (especially the continuity and semigroup property) in Definition \ref{dfn:control system} are satisfied, and as long as the derivative $\Dot{V}_u$ defined in Definition \ref{dfn:(L)ISS-Lyapunov function} exists for almost all $t\geq0$, then we can exploit (L)ISS-Lyapunov functions to study the ISS property.
In the existing literature, the notion of mild solution, defined using $C_0$ semigroups, is more commonly used \cite{dashkovskiy2013input}.
It however may lose many useful structures and properties of the specific equation (especially PDEs) under study, and requires more complicated techniques to characterize time-varying and nonlinear systems.
The notion of weak solution adopted in this work is standard for parabolic PDEs in the PDE literature, which satisfies the properties in Definition \ref{dfn:control system} when it exists and is unique.
(In fact, for linear PDEs with time-independent coefficients, these two notions are equivalent; see page 105 in \cite{curtain1995introduction}.)
By using weak solutions, we are able to obtain the necessary solution properties for studying the ISS property of system \eqref{eq:perturbed system} even with time-varying coefficients and unbounded control operators, which could have been difficult to study if using mild solutions.
\end{remark}
}

\begin{theorem} \label{thm:ISS}
(\textbf{LISS}). Consider the PDE system \eqref{eq:FP equation} with control law \eqref{eq:density feedback law using estimation}. Assume the regularity conditions \eqref{eq:regularity condition1} and \eqref{eq:regularity condition2} in Theorem \ref{thm:well-posedness} are satisfied and $p_0>0$. Define
\begin{equation} \label{eq:external input}
    d(t):=\max\left\{\left\|\frac{\nabla\epsilon}{1+\epsilon}\right\|_{L^\infty(\omega)}(t),\left\|\frac{\epsilon}{1+\epsilon}\right\|_{L^2(\omega)}(t)\right\}
\end{equation}
Then $d(t)=0$ if and only if $\epsilon(x,t)=0,\forall x$, and $\Phi$ is LISS in $L^2$ with respect to $d$ when
\begin{equation} \label{eq1:local bound for epsilon}
     \left\|\frac{\nabla\epsilon}{1+\epsilon}\right\|_{L^\infty(\omega)} < \frac{\alpha_{\text{min}}\theta}{C\|\alpha-\sigma\|_{L^\infty(\omega)}},
\end{equation}
where $C>0$, $\theta\in(0,1)$ are constants, and $\alpha_{\text{min}}(t):=\inf_{x\in{\color{black}\omega}}\alpha(x,t)$.
\end{theorem}


\begin{proof}
Since $\Phi$ is a weak solution of \eqref{eq:system of Phi}, according to \eqref{eq:energy identity}, we have the following energy identity:
\begin{align}\label{eq:Phi energy identity}
\begin{split}
    &\frac{1}{2}\int_{{\color{black}\omega}}\Phi^2dx-\frac{1}{2}\int_{{\color{black}\omega}}{\color{black}\Phi_0}^2dx\\
    &\quad=\int_0^t\int_{{\color{black}\omega}}\nabla\Phi\cdot\Big[{v}(\Phi+p^*) - \nabla\big(\sigma(\Phi+p^*)\big)\Big]dxd\tau.
\end{split}
\end{align}
Consider an LISS Lyapunov function $V(t)=\frac{1}{2}\|\Phi\|_{L^2({\color{black}\omega})}^2$.
Then
$$
V(t)-V(0)=\int_0^t\int_{{\color{black}\omega}}\nabla\Phi\cdot\Big[{v}(\Phi+p^*) - \nabla\big(\sigma(\Phi+p^*)\big)\Big]dxd\tau.
$$
Hence, $V$ is absolutely continuous on $[0,t]$ and, for almost every $t\in[0,T]$,
\begin{align}\label{eq:V derivative}
\begin{split}
    \Dot{V}(t)&=\int_{{\color{black}\omega}}\nabla\Phi\cdot\Big[{v}(\Phi+p^*) - \nabla\big(\sigma(\Phi+p^*)\big)\Big]dx\\
    &=\int_{{\color{black}\omega}}\nabla\Phi\cdot[{v}p - \nabla(\sigma p)]dx
\end{split}
\end{align}
Now substitute the control law \eqref{eq:density feedback law using estimation} into \eqref{eq:V derivative}, and use $\Hat{p}=p(1+\epsilon)$ and $p=\Phi+p^*$. 
We have
\begin{align*}
    \Dot{V} 
    &= - \int_{\omega} \nabla\Phi\cdot\Big( \frac{p[\alpha\nabla(\Hat{p}-p^*)-\nabla(\sigma\Hat{p})]}{\Hat{p}} + \nabla(\sigma p) \Big) dx\\
    &= \int_{\omega}-\nabla \Phi \cdot \Big(\frac{\alpha \nabla [\Phi(1+\epsilon)] + \alpha \nabla(\epsilon p^*) - \sigma(\Phi+p^*) \nabla\epsilon}{1+\epsilon} \Big) dx\\
    &= \int_{\omega}-\alpha|\nabla \Phi|^{2}-\nabla \Phi \cdot \frac{(\alpha-\sigma)\left(\Phi+p^{*}\right) \nabla \epsilon+\alpha \epsilon \nabla p^{*}}{1+\epsilon} dx\\
    &\leq \int_{\omega}-\alpha|\nabla \Phi|^{2} + \left|\frac{(\alpha-\sigma)\Phi\nabla\Phi\cdot\nabla \epsilon}{1+\epsilon}\right|\\
    &\quad + \left|\frac{(\alpha-\sigma)p^{*}\nabla\Phi\cdot\nabla \epsilon}{1+\epsilon}\right| + \left|\frac{\alpha\epsilon\nabla\Phi\cdot\nabla p^{*}}{1+\epsilon}\right| dx.
\end{align*}
Let $\alpha_{\text{min}}(t):=\inf_{x\in{\color{black}\omega}}\alpha(x,t)>0$, choose a constant $\theta\in(0,1)$ to split the first term into two terms, and apply the H\"older's inequality for the remaining terms. Then we have
\begin{equation*} \label{eq:proof}
    \begin{aligned}
    \Dot{V} 
    &\leq -\alpha_{\text{min}}(1-\theta)\|\nabla\Phi\|_{L^2(\omega)}^2 -\alpha_{\text{min}}\theta \|\nabla\Phi\|_{L^2(\omega)}^2 \\
    &\quad + \|\nabla\Phi\|_{L^2(\omega)}\|\Phi\|_{L^2(\omega)}\|\alpha-\sigma\|_{L^\infty(\omega)}\left\|\frac{\nabla\epsilon}{1+\epsilon}\right\|_{L^\infty(\omega)}\\
    &\quad + \|\nabla\Phi\|_{L^2(\omega)}\|p^*\|_{L^2(\omega)}\|\alpha-\sigma\|_{L^\infty(\omega)}\left\|\frac{\nabla\epsilon}{1+\epsilon}\right\|_{L^\infty}\\
    &\quad + \|\nabla\Phi\|_{L^2(\omega)}\|\alpha\nabla p^*\|_{L^\infty(\omega)}\left\|\frac{\epsilon}{1+\epsilon}\right\|_{L^2(\omega)} \\
    &\qquad \text{(by the Poincar\'e inequality)} \\
    &\leq -\frac{\alpha_{\text{min}}(1-\theta)}{C^2} \|\Phi\|_{L^2(\omega)}^2
    - \frac{\alpha_{\text{min}}\theta}{C} \|\nabla\Phi\|_{L^2(\omega)} \|\Phi\|_{L^2(\omega)}\\
    &\quad + \|\nabla\Phi\|_{L^2(\omega)}\|\Phi\|_{L^2(\omega)}\|\alpha-\sigma\|_{L^\infty(\omega)}\left\|\frac{\nabla\epsilon}{1+\epsilon}\right\|_{L^\infty(\omega)}\\
    &\quad + \|\nabla\Phi\|_{L^2(\omega)}\|p^*\|_{L^2(\omega)}\|\alpha-\sigma\|_{L^\infty(\omega)}\left\|\frac{\nabla\epsilon}{1+\epsilon}\right\|_{L^\infty(\omega)}\\
    &\quad + \|\nabla\Phi\|_{L^2(\omega)}\|\alpha\nabla p^*\|_{L^\infty(\omega)}\left\|\frac{\epsilon}{1+\epsilon}\right\|_{L^2(\omega)}.
\end{aligned}
\end{equation*}
Thus, we would have
$$
\Dot{V}\leq -\frac{\alpha_{\text{min}}(1-\theta)}{C^2}\|\Phi\|_{L^2(\omega)}^2 =:-W(\|\Phi\|_{L^2(\omega)}),
$$
if
\begin{equation} \label{eq2:local bound for epsilon}
    \begin{aligned} 
    \frac{\alpha_{\text{min}}\theta}{C} \|\Phi\|_{L^2(\omega)} 
    &\geq \|\Phi\|_{L^2(\omega)}\|\alpha-\sigma\|_{L^\infty(\omega)}\left\|\frac{\nabla\epsilon}{1+\epsilon}\right\|_{L^\infty(\omega)}\\
    &\quad + \|p^*\|_{L^2(\omega)}\|\alpha-\sigma\|_{L^\infty(\omega)}\left\|\frac{\nabla\epsilon}{1+\epsilon}\right\|_{L^\infty(\omega)}\\
    &\quad + \|\alpha\nabla p^*\|_{L^\infty(\omega)}\left\|\frac{\epsilon}{1+\epsilon}\right\|_{L^2(\omega)}.
\end{aligned}
\end{equation}
Inequality \eqref{eq2:local bound for epsilon} holds if
\begin{equation*}
    \frac{\alpha_{\text{min}}\theta}{C} > \|\alpha-\sigma\|_{L^\infty(\omega)}\left\|\frac{\nabla\epsilon}{1+\epsilon}\right\|_{L^\infty(\omega)},
\end{equation*}
and 
\begin{align} \label{eq:Phi greater equal}
\begin{split}
    \|\Phi\|_{L^2(\omega)} 
    \geq \frac{
    \displaystyle\|p^*\|_{L^2(\omega)}\|\alpha-\sigma\|_{L^\infty(\omega)}\left\|\frac{\nabla\epsilon}{1+\epsilon}\right\|_{L^\infty(\omega)} 
    }{
    \displaystyle\frac{\alpha_{\text{min}}\theta}{C} - \|\alpha-\sigma\|_{L^\infty(\omega)}\left\|\frac{\nabla\epsilon}{1+\epsilon}\right\|_{L^\infty(\omega)}
    }\\
    +\frac{
    \displaystyle\|\alpha\nabla p^*\|_{L^\infty(\omega)}\left\|\frac{\epsilon}{1+\epsilon}\right\|_{L^2(\omega)}
    }{
    \displaystyle\frac{\alpha_{\text{min}}\theta}{C} - \|\alpha-\sigma\|_{L^\infty(\omega)}\left\|\frac{\nabla\epsilon}{1+\epsilon}\right\|_{L^\infty(\omega)}
    } .
\end{split}
\end{align}
With $d(t)=\max\left\{\left\|\frac{\nabla\epsilon}{1+\epsilon}\right\|_{L^\infty(\omega)}(t),\left\|\frac{\epsilon}{1+\epsilon}\right\|_{L^2(\omega)}(t)\right\}$, we obtain that \eqref{eq:Phi greater equal} holds if
\begin{align*}
    \|\Phi\|_{L^2(\omega)} 
    \geq \frac{\displaystyle\|p^*\|_{L^2(\omega)}\|\alpha-\sigma\|_{L^\infty(\omega)}d + \|\alpha\nabla p^*\|_{L^\infty(\omega)}d } {\displaystyle\frac{\alpha_{\text{min}}\theta}{C} - \|\alpha-\sigma\|_{L^\infty(\omega)}d} =:\chi(d)
\end{align*}
Since $W$ is positive definite and $\chi\in\mathcal{K}$, according to Theorem \ref{thm:(L)ISS-Lyapunov function}, we obtain the LISS property.
\end{proof}

Theorem \ref{thm:ISS} indicates that the closed-loop system using control law \eqref{eq:density feedback law using estimation} remains bounded as long as the density estimation error satisfies the constraint \eqref{eq1:local bound for epsilon}. Note that we can drop the constraint \eqref{eq1:local bound for epsilon} and obtain (global) ISS property if we let $\alpha=\sigma$, in the price of possibly slow convergence since $\sigma$ is usually small. Therefore, the design parameter $\alpha$ yields a trade-off between convergence speed and robustness in the sense that a larger $\alpha$ produces faster convergence but also reduces the convergent domain in \eqref{eq1:local bound for epsilon}. Nevertheless, one can also increase the probability of satisfying \eqref{eq1:local bound for epsilon} by increasing the number of robots $N$, which is the consistency property of KDE.

\begin{remark}
The term $\nabla\epsilon$ in \eqref{eq:external input} is caused by the gradient operator $\nabla$ in \eqref{eq:density feedback law using estimation}, which is unavoidable because the gradient operator is unbounded, that is, we cannot bound $\nabla\epsilon$ using $\epsilon$. In fact, $\epsilon$ also acts on the system \eqref{eq:system of Phi} through the divergence operator $\nabla\cdot$ on the right-hand side. The reason why it does not show up in \eqref{eq:external input} is that in the formulation of weak solutions \eqref{eq:weak solution}, by using integration by parts, the divergence actually acts on the test functions $\eta$ (corresponding to $\Phi$ in \eqref{eq:Phi energy identity}) and produces a boundary term on $S(\Omega)$ which eventually disappears due to the reflecting boundary condition in \eqref{eq:FP equation}.
{\color{black}
It would have been difficult to deal with the unbounded divergence operator if we adopt mild solutions for ISS analysis.
}
\end{remark}

\begin{remark} \label{remark:distributed KDE}
{\color{black}
We shall clarify that our swarm control strategy essentially consists of two parts: centralized velocity field design and distributed velocity tracking.
This work mainly focuses on the design of velocity field, which is centralized because it requires knowing the positions of all the robots to estimate their density.
This can be implemented by a monitoring system which collects the robots' positions to estimate the global density and then broadcasts the velocity field to the robots. 
The individual velocity tracking control is however distributed because each robot receives its reference velocity command and then derive its own control signal accordingly (which is a well-studied control problem for mobile robots).
}
\end{remark}
}

\section{Simulation studies} \label{section:simulation}
An agent-based simulation using 1024 robots is performed on Matlab to verify the proposed control law. We set $\omega=(0,1)^2$, $\sigma=0.0005$ and $\alpha=0.03$. 
{\color{black}
Each robot is simulated by a Langevin equation \eqref{eq:Langevin equation} under the velocity command \eqref{eq:density feedback law using estimation}. The robots' initial positions are drawn from a uniform distribution. The desired density $p^*(x)$ is illustrated in Fig. \ref{fig:desired density} (which is $C^\infty$ and lower bounded by a very small positive constant due to smoothing preprocessing). 
}
KDE is used to obtain the density estimate $\Hat{p}(x,t)$, in which we set $h=0.045$. Numerical computation of the velocity field \eqref{eq:density feedback law using estimation} is based on finite difference. Specifically, $\omega$ is discretized into a $64\times64$ grid, and the time difference is $0.02s$. 

\begin{figure}[hbt!]
\setlength{\abovecaptionskip}{0.0cm}
\setlength{\belowcaptionskip}{-0.0cm}
    \centering
    \begin{subfigure}[b]{0.24\textwidth}
        \centering
        \includegraphics[width=\textwidth]{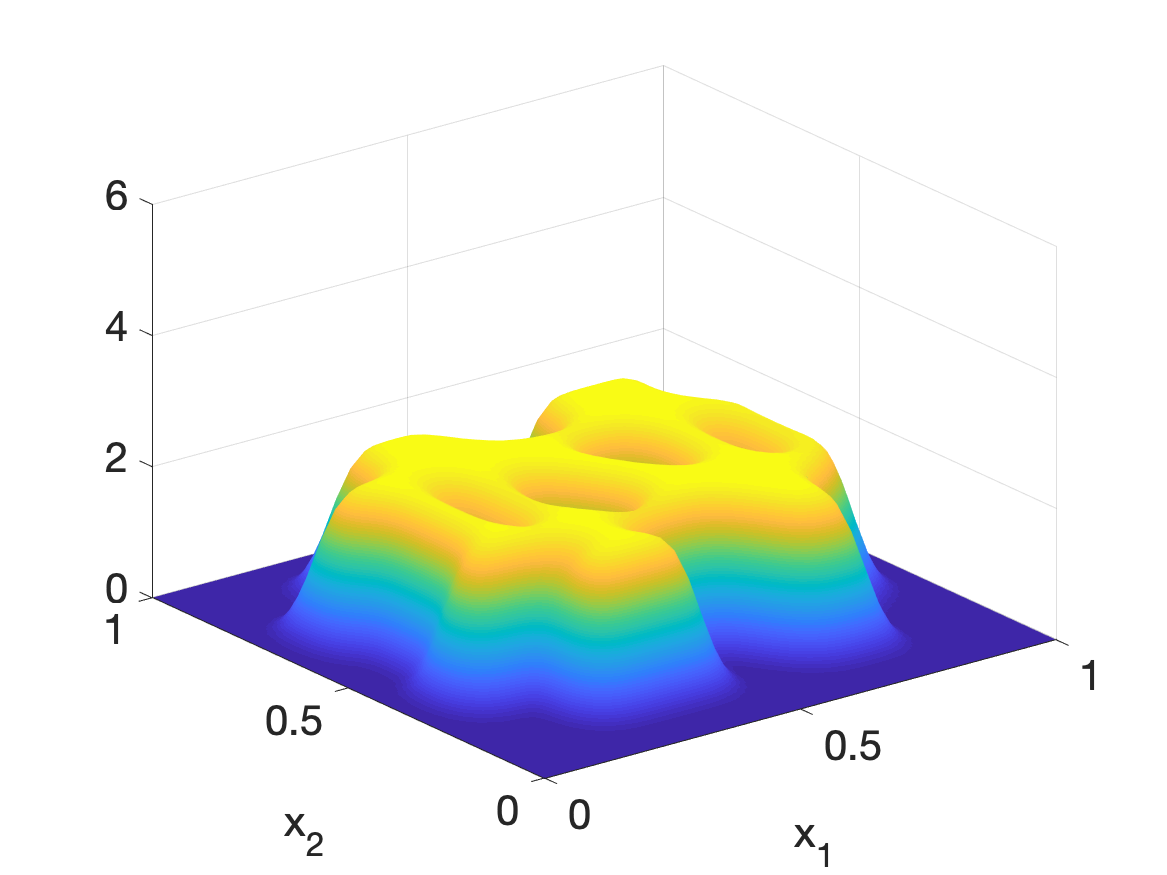}
        \caption{The desired density $p^*(x)$.}
        \label{fig:desired density}
    \end{subfigure}
    \begin{subfigure}[b]{0.24\textwidth}
        \centering
        \includegraphics[width=\textwidth]{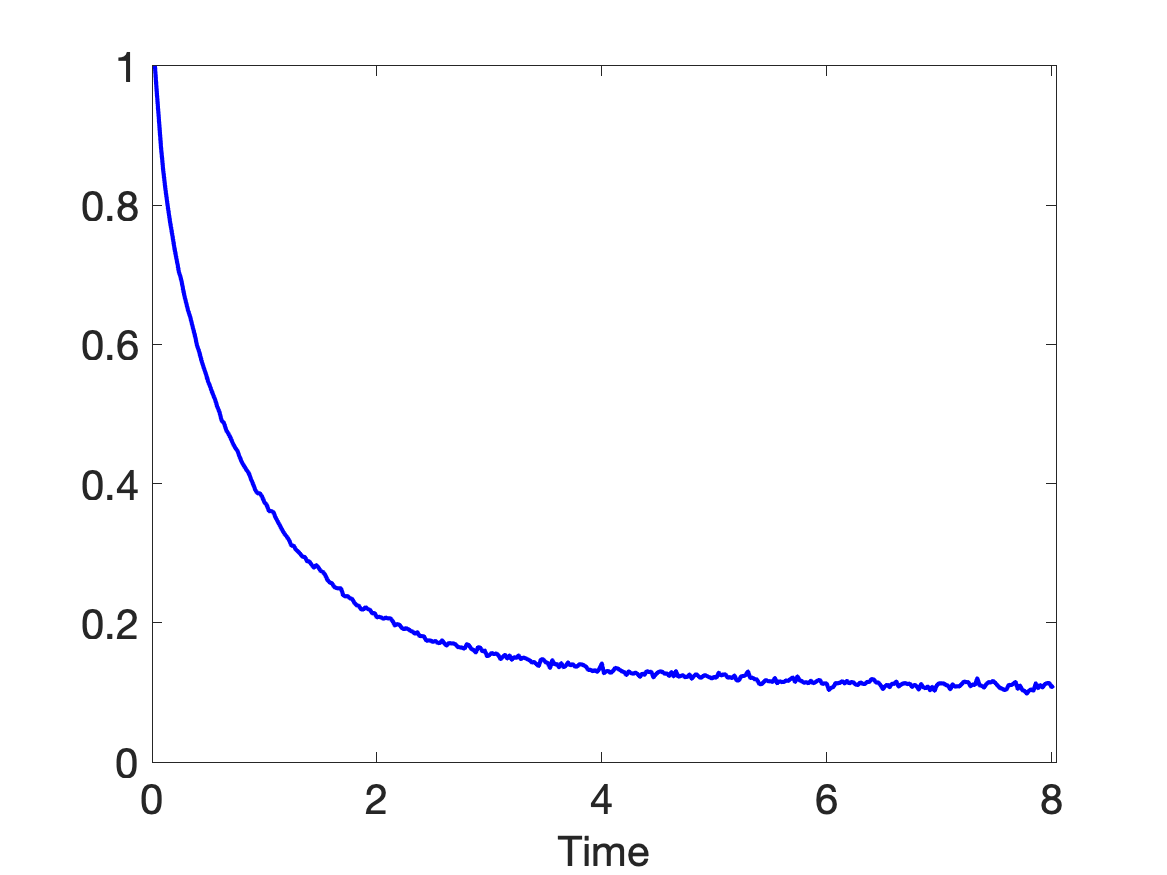}
        \caption{The convergence error.}
        \label{fig:convergence error}
    \end{subfigure}
    \caption{}
\end{figure}

\begin{figure*}[t]
\setlength{\abovecaptionskip}{0.0cm}
\setlength{\belowcaptionskip}{-0.5cm}
    \centering
    \begin{subfigure}[b]{0.24\textwidth}
        \centering
        \includegraphics[width=\textwidth]{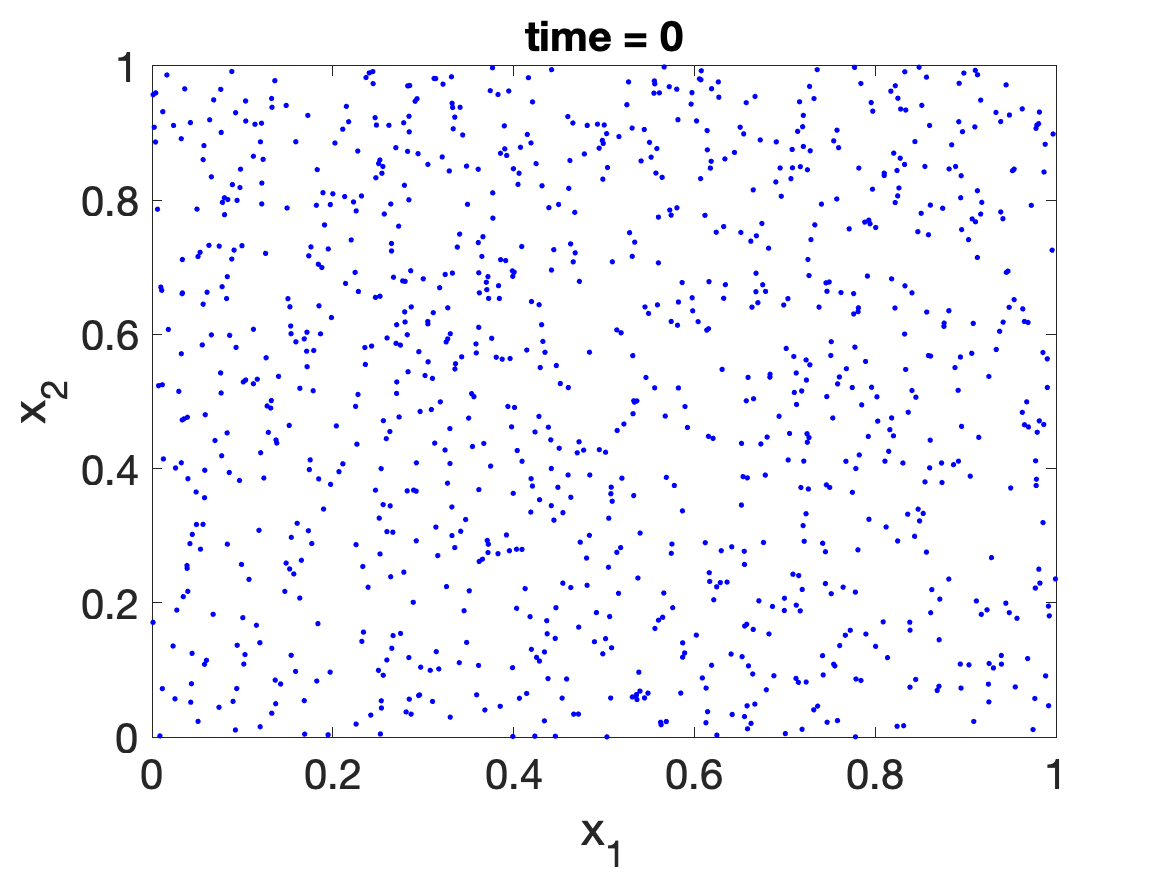}
    \end{subfigure}
    \begin{subfigure}[b]{0.24\textwidth}
        \centering
        \includegraphics[width=\textwidth]{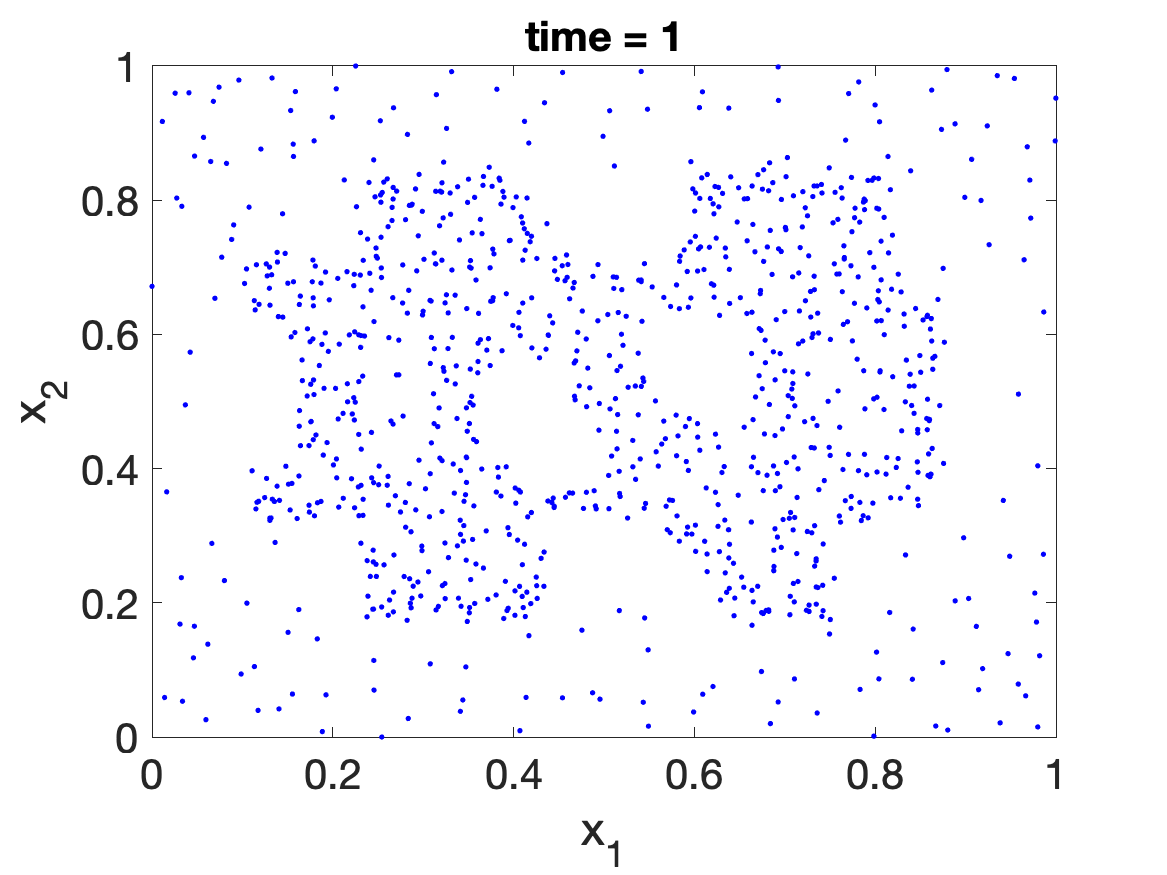}
    \end{subfigure}
    \begin{subfigure}[b]{0.24\textwidth}
        \centering
        \includegraphics[width=\textwidth]{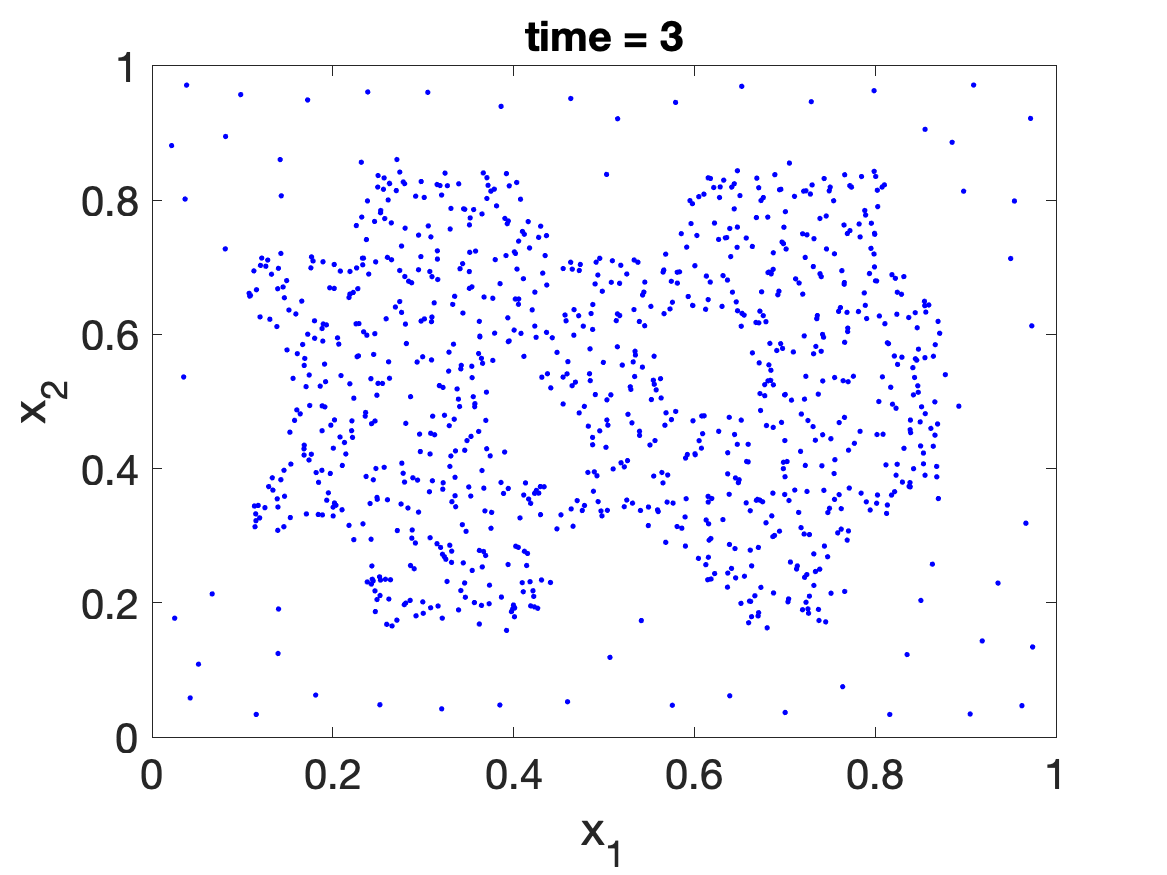}
    \end{subfigure}
    \begin{subfigure}[b]{0.24\textwidth}
        \centering
        \includegraphics[width=\textwidth]{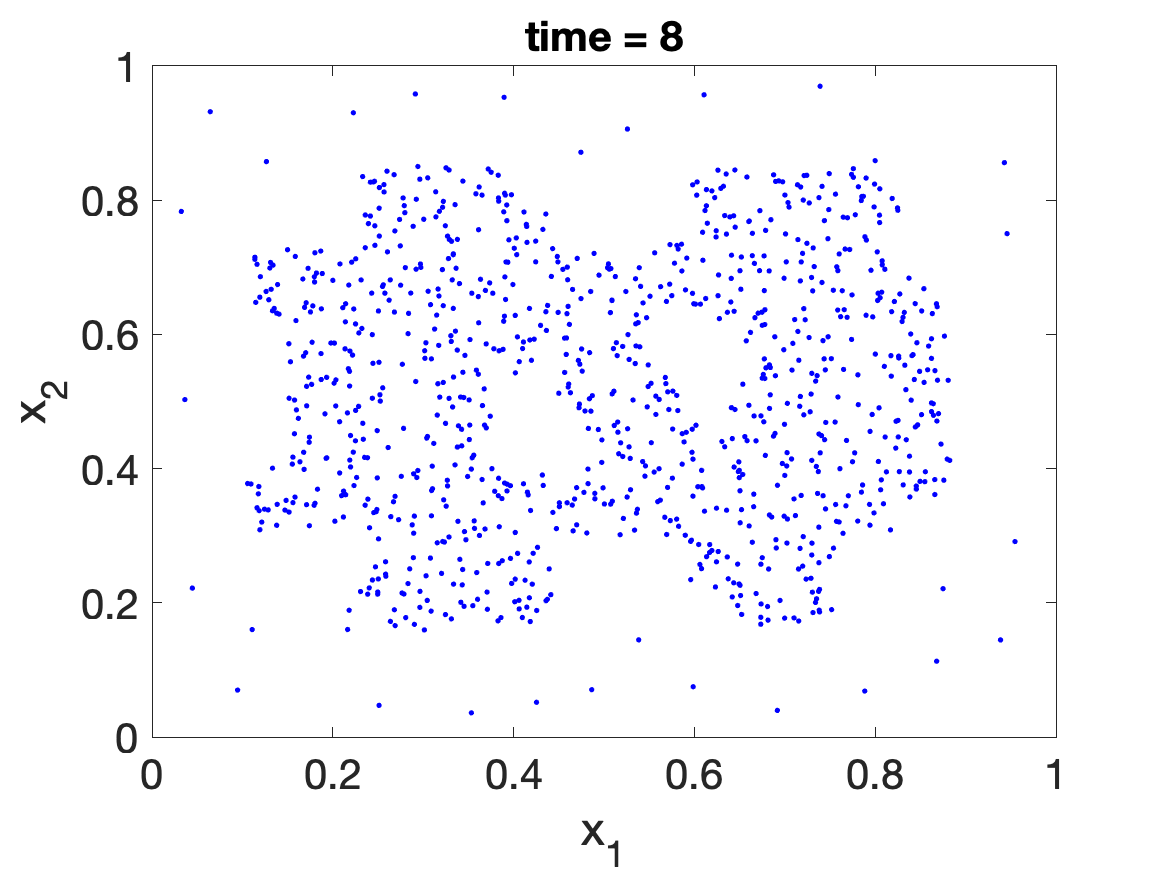}
    \end{subfigure}
    
    \begin{subfigure}[b]{0.24\textwidth}
        \centering
        \includegraphics[width=\textwidth]{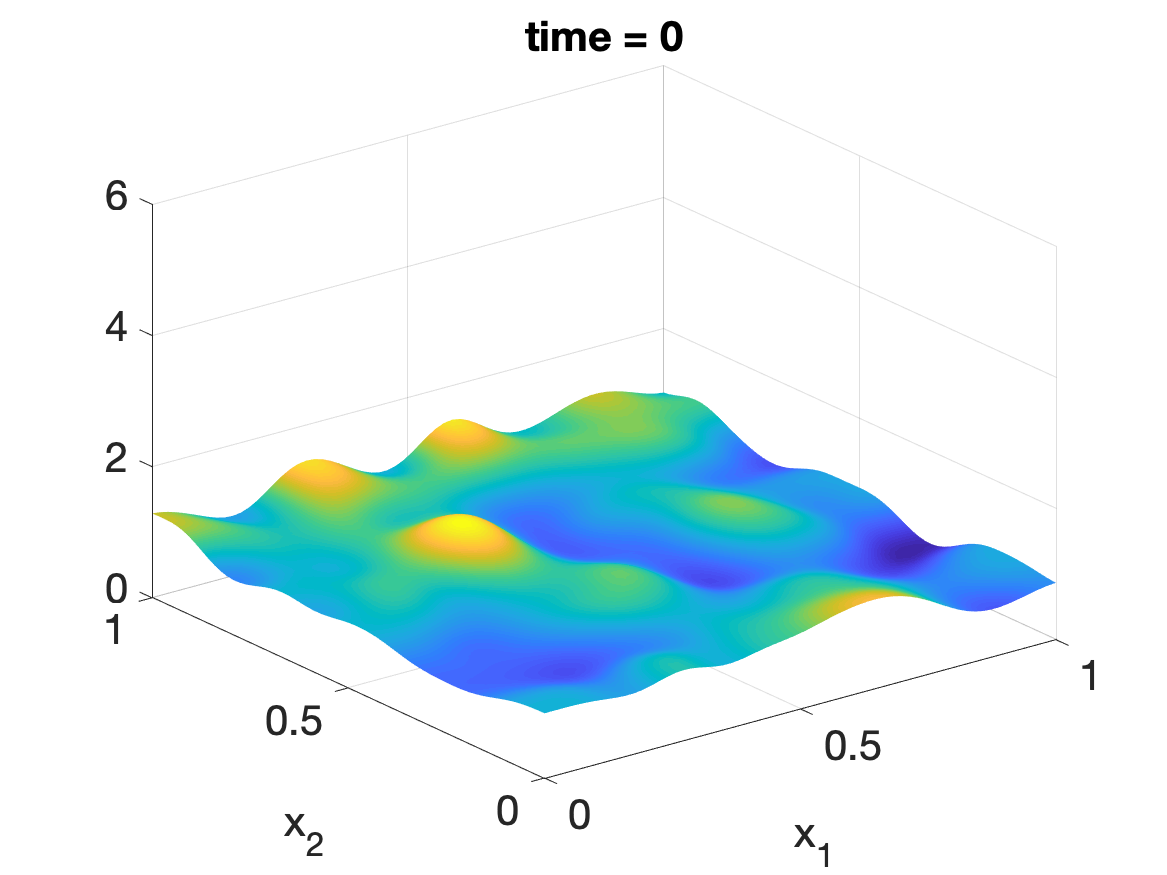}
    \end{subfigure}
    \begin{subfigure}[b]{0.24\textwidth}
        \centering
        \includegraphics[width=\textwidth]{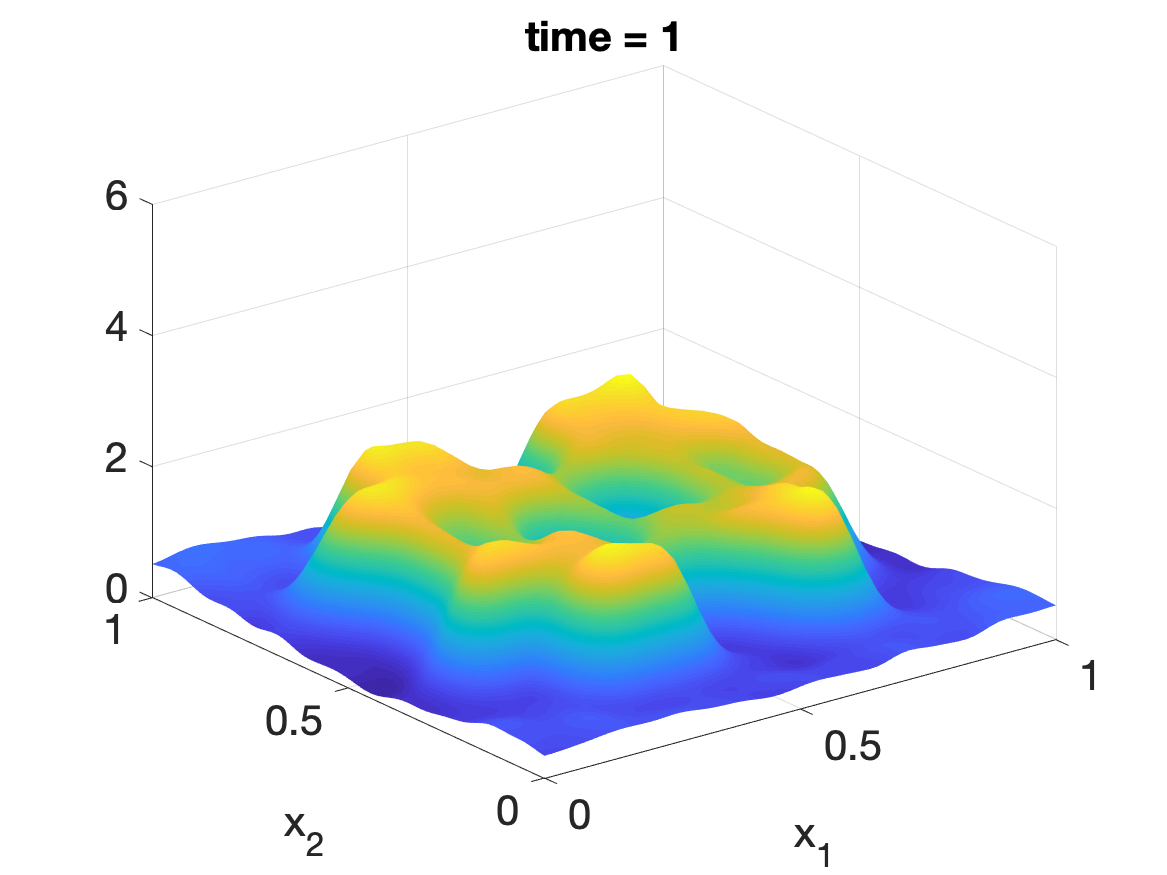}
    \end{subfigure}
    \begin{subfigure}[b]{0.24\textwidth}
        \centering
        \includegraphics[width=\textwidth]{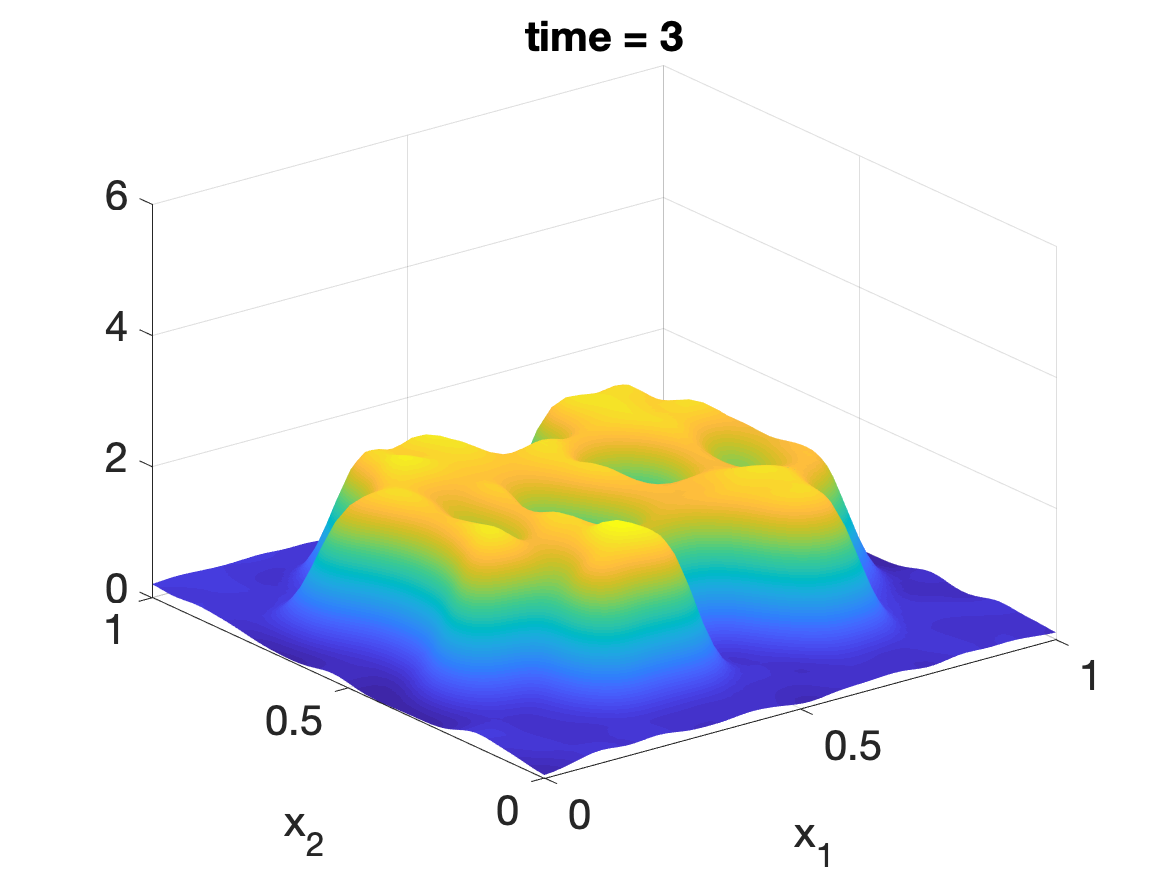}
    \end{subfigure}
    \begin{subfigure}[b]{0.24\textwidth}
        \centering
        \includegraphics[width=\textwidth]{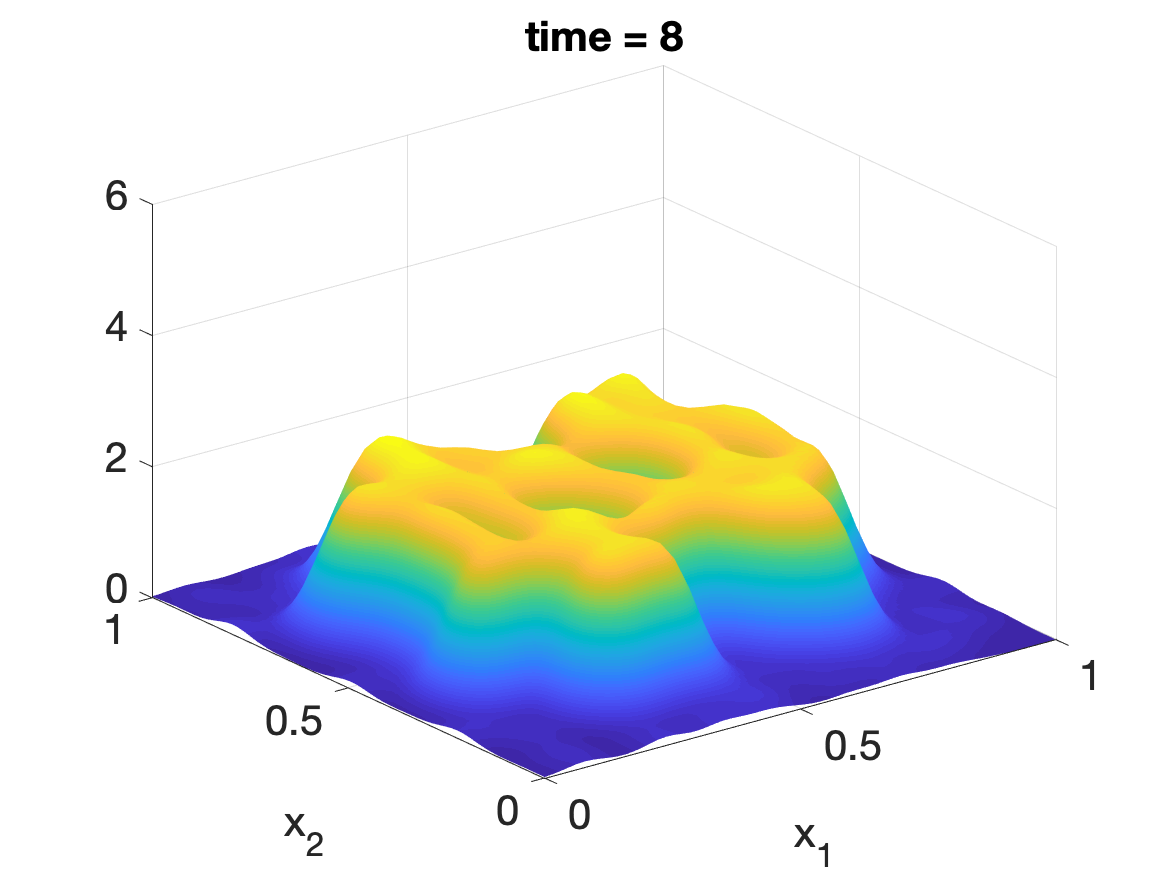}
    \end{subfigure}
    
    \begin{subfigure}[b]{0.24\textwidth}
        \centering
        \includegraphics[width=\textwidth]{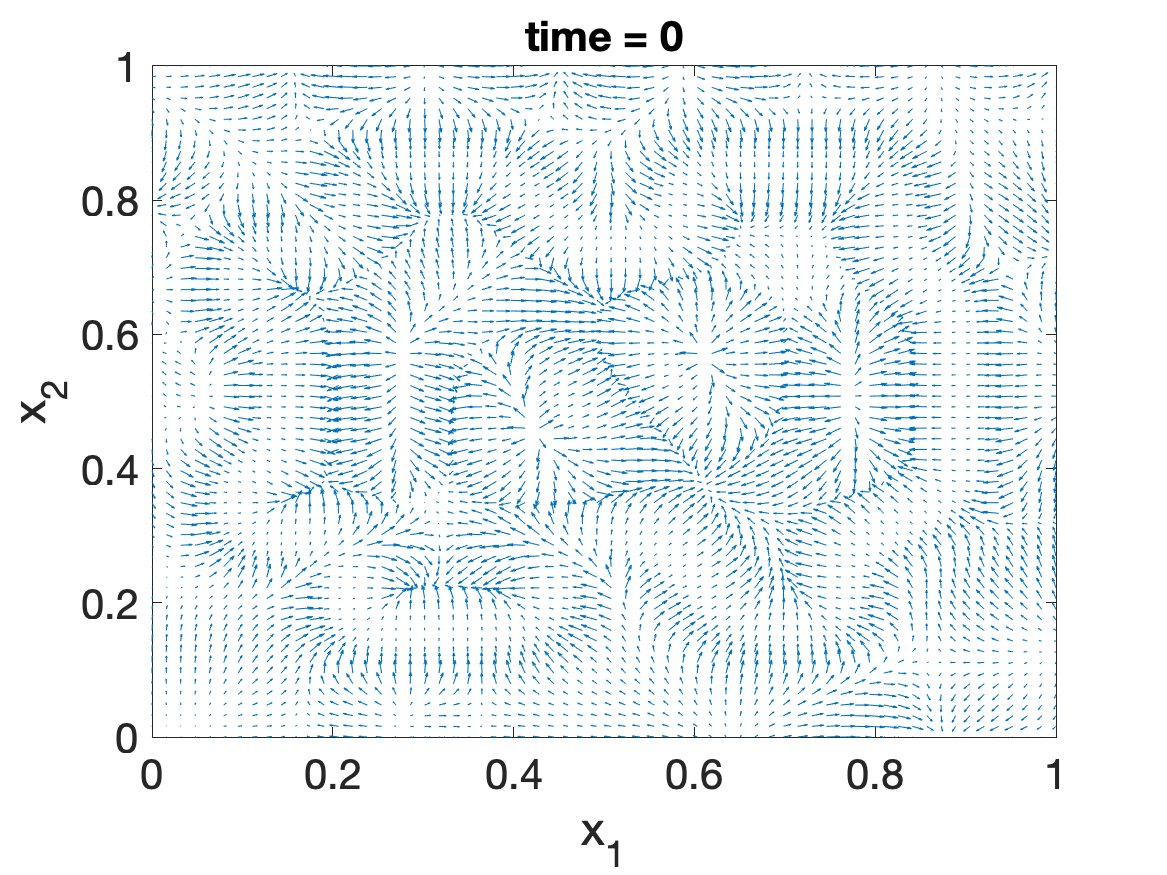}
    \end{subfigure}
    \begin{subfigure}[b]{0.24\textwidth}
        \centering
        \includegraphics[width=\textwidth]{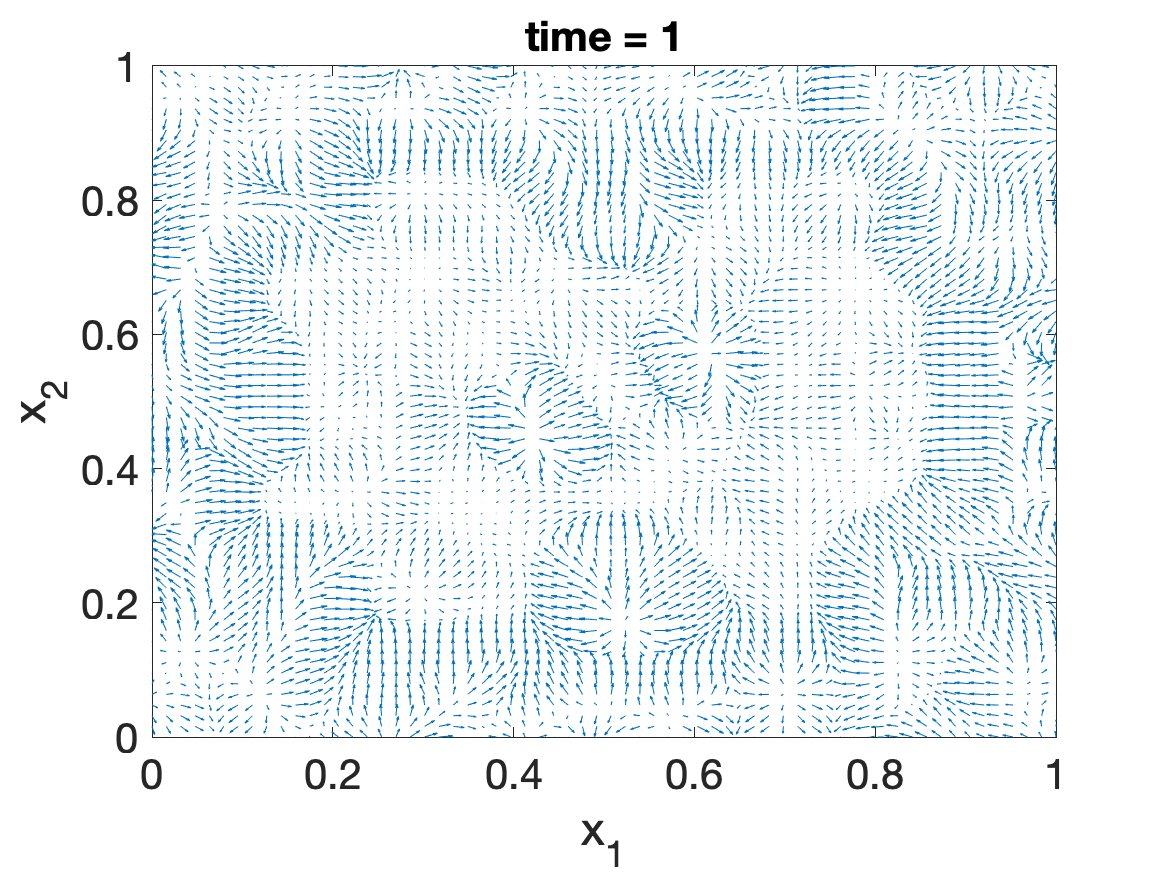}
    \end{subfigure}
    \begin{subfigure}[b]{0.24\textwidth}
        \centering
        \includegraphics[width=\textwidth]{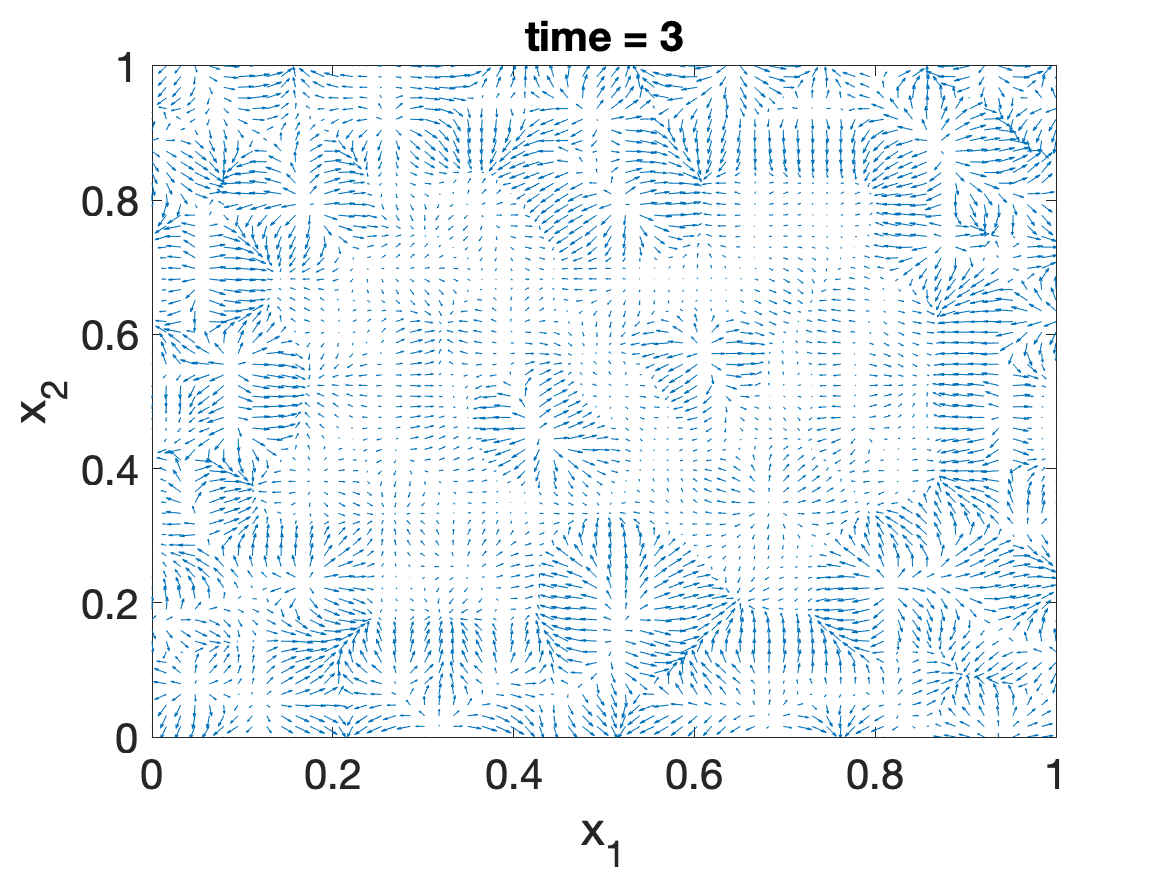}
    \end{subfigure}
    \begin{subfigure}[b]{0.24\textwidth}
        \centering
        \includegraphics[width=\textwidth]{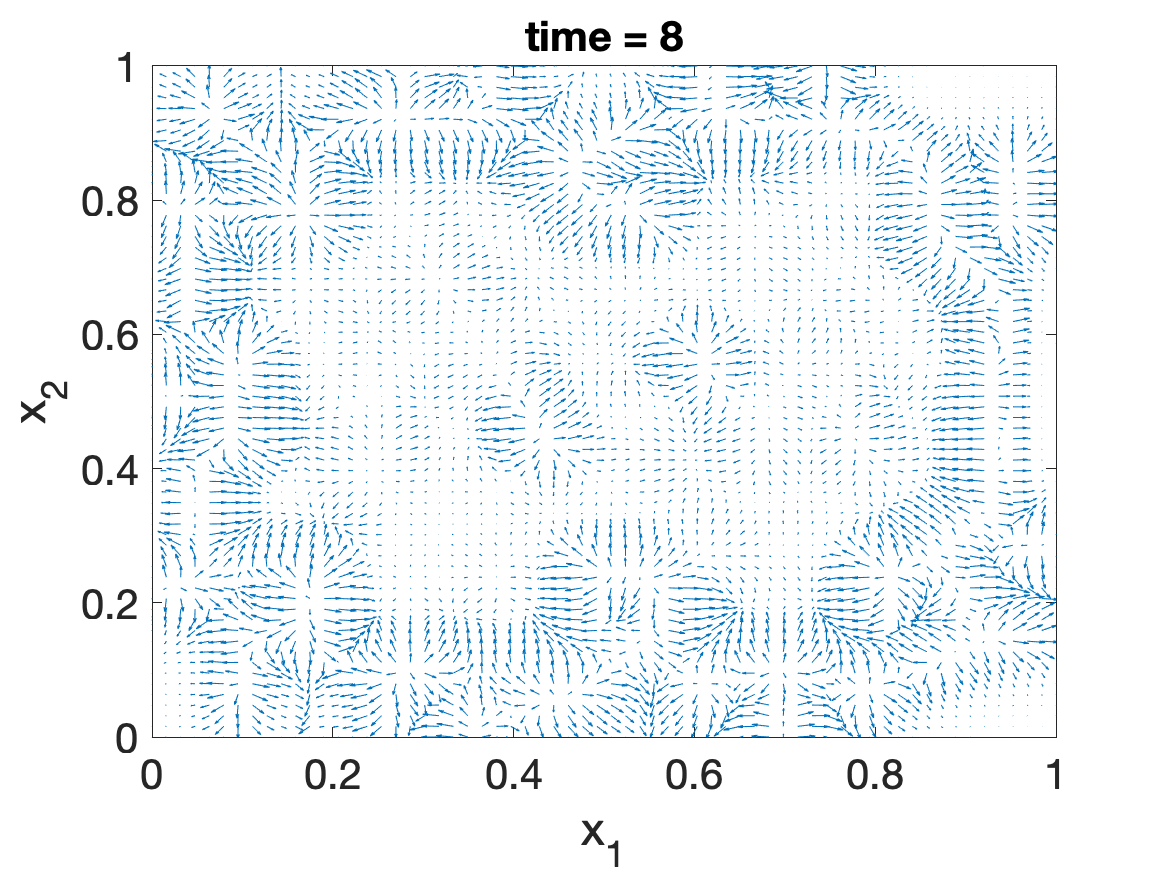}
    \end{subfigure}
    \caption{Evolution of the swarm (top), the density estimates $\Hat{p}(x,t)$ (middle) and the real-time velocity fields ${v}(x,t)$ (bottom). Magnitude of the velocity fields is rescaled for illustrative purpose.}
    \label{fig:pdf evolution}
\end{figure*}

Fig. \ref{fig:pdf evolution} demonstrates the positions of the robots $\{X_i(t)\}_{i=1}^N$, the estimated density $\Hat{p}(x,t)$ of the swarm, and the velocity field ${v}(x,t)$ generated by \eqref{eq:density feedback law using estimation}, which suggests that the swarm is able to evolve towards the desired configuration. The convergence error $\|\Hat{p}-p^*\|_{L^2(\omega)}$ is given in Fig. \ref{fig:convergence error}, which shows that the error converges exponentially to a small neighbourhood around $0$ and remains bounded, which verifies the ISS property of the proposed algorithm.

\section{Conclusions} \label{section:conclusion}
This paper studied controlling the density of a swarm of robots using velocity fields that are computed in a feedback manner. The resulting closed-loop system was proven to be LISS with respect to density estimation errors.
The presented framework filled the gap between local kinematics of individual robots and their emergent behaviors in swarm robotic systems. It was top-down and computationally efficient.
With the feedback technique, the global performance was guaranteed to be convergent and robust to estimation errors when performing in real-time. 
Our future work includes studying the distributed density estimation problem and considering more general robotic dynamics.


\section*{Appendices}
{\color{black}
\subsection{Conormal derivative problems}\label{section:conormal derivative problems}
Equation \eqref{eq:FP equation} is a special case of the so-called \textit{conormal derivative problem} for parabolic equations of divergent form \cite{lieberman1996second}. We summarize (and modify appropriately) main results from Chapter VI \cite{lieberman1996second}.

We use the same notations as in Section \ref{section:notation}. We follow the summation convention that any term with a repeated index $i$ is summed over $i=1$ to $n$. For example, $b_{i} \partial_{i} u=\sum_{i=1}^{n} b_{i} \partial_{i} u$. For bounded functions $a_{ij}$, $b_i$, $c_i$, and $c_0$ in $\omega$, define the operator
$$
Lu:=-\partial_tu+\partial_i(a_{ij}\partial_ju+b_iu)+c_i\partial_iu+c_0u.
$$
We always assume that $\{a_{ij}\}$ is uniformly elliptic, i.e., for some positive constant $\lambda$, 
$$
a_{ij}(x,t)\xi_i\xi_j\ge\lambda|\xi|^2\quad\text{for any $(x,t)\in\Omega$ and any $\xi\in\mathbb R^n$}.
$$
For a bounded function $b_0$ on $S(\Omega)$, define the operator 
$$
Mu:=(a_{ij}\partial_ju+b_iu-f_i)\nu_i - b_0u\quad\text{on }S(\Omega),
$$
where $\nu=(\nu_1, \cdots, \nu_n)$ is the unit inner normal to the boundary.

For given functions $f_i$ and $g$ on $\Omega$, $\varphi$ on $\omega$, and 
$\psi$ on $S(\Omega)$, the {\it conormal derivative problem} has the following form: 
\begin{align} \label{eq:conormal problem}
\begin{split}
Lu&=\partial_if_i+g\quad\text{in }\Omega,\\
u&=\varphi\quad\text{on }{\color{black}\omega\times\{0\}},\\
Mu&=\psi\quad\text{on }S(\Omega).
\end{split}
\end{align}

In this paper, we only need the case $f_i=g=0$ in $\Omega$ and $b_0=\psi=0$ on $S(\Omega)$. We present the general form for completeness.
Take any test function $\eta\in C^1(\overline{\Omega})$. 
Multiplying the first equation of \eqref{eq:conormal problem} by $-\eta$ and integrating by parts, we obtain 
\begin{align} \label{eq:weak solution}
\begin{split}
&\int_{{\color{black}\omega}}u \eta dx-\int_{\Omega}u\partial_t\eta dxd\tau\\
&\quad+\int_{\Omega}(a_{ij}\partial_ju+b_iu-f_i)\partial_i\eta-(c_i\partial_iu+c_0u-g)\eta dxd\tau\\
&\quad=\int_{S(\Omega)}(b_0u+\psi)\eta ds d\tau+\int_{{\color{black}\omega}}\varphi \eta dx,
\end{split}
\end{align}
where $ds$ is the area form of the boundary $\partial\omega$.

In the following, we always assume $a_{ij}, b_i, c_i, c_0\in L^\infty(\Omega)$, and $b_0\in L^\infty(S(\Omega))$.
We also consider given $f_i, g\in L^2(\Omega)$, $\varphi\in L^2(\omega)$, and $\psi\in L^2(S(\Omega))$. 
For convenience, we write $f=(f_1, \cdots, f_n)$. 

\begin{definition} \label{dfn:weak solution}
(Weak solution \cite{lieberman1996second}). A function $u\in \mathcal{M}$ is a {\it weak solution} of the initial/boundar-value problem \eqref{eq:conormal problem} if it satisfies \eqref{eq:weak solution} for any $\eta\in H^1(\Omega)$ and almost every $t\in(0,T]$. Similarly, a function $u\in\mathcal{M}$ is a {\it weak subsolution (supersolution)} of the problem \eqref{eq:conormal problem} if the inequality $\leq (\geq)$ holds in \eqref{eq:weak solution} instead of the equality $=$, 
for any $\eta\in H^1(\Omega)$ with $\eta\geq 0$ and almost every $t\in(0,T]$. 
\end{definition}


We note that a weak solution is simultaneously a weak subsolution and a weak supersolution.
We now discuss the well-posedness and some properties of the weak solution.
The following result is based on Theorem 6.38 and Theorem 6.39 in \cite{lieberman1996second}.
\begin{theorem}\label{thm:well-posedness of conormal}
(\textbf{Well-posedness} \cite{lieberman1996second}). Assume $f_i, g\in L^2(\Omega)$, $\varphi\in L^2(\omega)$, and $\psi\in L^2(S(\Omega))$. 
Then, there exists a unique weak solution $u\in \mathcal{M}$ of the problem \eqref{eq:conormal problem}, which satisfies
$$
\|u\|_\mathcal{M}\leq Ce^{CT}\{\|f\|_{L^2(\Omega)}+\|g\|_{L^2(\Omega)}+\|\varphi\|_{L^2(\omega)}+\|\psi\|_{L^2(S(\Omega))}\},
$$ 
where $C$ is a positive constant depending only on $n$, $\lambda$, $\omega$, and the $L^\infty$-norms of $a_{ij}$, $b_i$, $c_i$, $c_0$, and $b_0$.
\end{theorem}

We have the following \textit{energy identity} for the weak solution $u\in \mathcal{M}$: for almost every $t\in(0,T]$,
\begin{align}
&\frac{1}{2}\int_{{\color{black}\omega}}u^2dx \nonumber\\
&\quad+\int_\Omega\big[(a_{ij}\partial_ju+b_iu-f_i)\partial_iu-(c_i\partial_iu+c_0u-g)u\big]dxd\tau \nonumber\\
&\quad=\int_{S(\Omega)}(b_0u+\psi)udsd\tau+\frac12\int_{{\color{black}\omega}}\varphi^2dx. \label{eq:energy identity}
\end{align}
The proof is by an approximation argument, i.e., take $\eta=u\in \mathcal{M}$ and show it is the limit in $\mathcal{M}$ of a sequence of $H^1$ functions \cite{lieberman1996second}. 

From now on, we assume $\omega$ is a connected domain. The following result is based on Theorem 6.43 in \cite{lieberman1996second}, which is for subsolutions.
\begin{theorem}\label{thm:strong maximum priciple}
(\textbf{Strong maximum principle} \cite{lieberman1996second}). Assume $f_i=g=0$ in $\Omega$, $\psi=0$ and $b_0\leq 0$ on $S(\Omega)$, $\varphi\in L^\infty(\omega)$, and, for any $v\in C^1(\Omega)$ with $v\geq 0$,
\begin{equation}\label{eq:sign coefficient}
    \int_{\Omega}(-b_i\partial_iv+c_0v)dxdt\leq 0.
\end{equation}
Let $u\in \mathcal{M}$ be a weak subsolution of the problem \eqref{eq:conormal problem}. Then, 
$$
u\geq-\sup_{\omega}\varphi^-\quad\text{in }{\color{black}\omega\times(0,T]}.
$$
Moreover, $u$ is constant if the equality holds at some $(x,t)\in {\color{black}\omega\times(0,T]}$. 

\end{theorem}

If $\partial_ib_i\in L^\infty(\Omega)$, the condition \eqref{eq:sign coefficient} can be 
substituted by its pointwise form $\partial_ib_i+c_0\leq 0$ in $\Omega$, and is not needed if we compare $u$ with 0.
Specifically, we have the following positivity result.

\begin{corollary}\label{corollary:positivity}
(\textbf{Positivity}). Assume  
$f_i=g=0$ in $\Omega$, $\psi=0$ and $b_0\leq 0$ on $S(\Omega)$, $\partial_ib_i\in L^\infty(\Omega)$, and $\varphi\in L^\infty(\omega)$. 
Let $u\in \mathcal{M}$ be a weak subsolution of the problem \eqref{eq:conormal problem}. 
If $ \varphi\geq(\text{or}>)0$ on $\omega$, then 
$$
u\geq(\text{or}>)0\quad\text{in }{\color{black}\omega\times(0,T]}.
$$
Moreover, $u$ is constant if the equality holds at some $(x,t)\in{\color{black}\omega\times(0,T]}$.
\end{corollary}
\begin{proof}
Consider $u=e^{\mu t}w$. Then, $w$ is a weak solution of the equation 
$(L-\mu)w=0$. The coefficient of the zero-order term is given by $c_0-\mu$. 
By taking $\mu\geq \partial_ib_i+c_0$, the pointwise version of \eqref{eq:sign coefficient} holds for the operator $L-\mu$. 
We may apply Theorem \ref{thm:strong maximum priciple} to $w$ to conclude 
$w\geq(\text{or}>) 0$ in ${\color{black}\omega\times(0,T]}$ since $\varphi\geq(\text{or}>) 0$ on $\omega$. 
Hence, $u\geq(\text{or}>) 0$ in ${\color{black}\omega\times(0,T]}$.
\end{proof}
}
{\color{black}

\subsection{Proof of Theorem \ref{thm:(L)ISS-Lyapunov function}}
\begin{proof}
Let the control system $\Sigma=\left(X, U_{c}, \phi\right)$ possess an LISS-Lyapunov function and $\psi_{1}, \psi_{2}, \chi, W, \rho_{x}, \rho_{u}$ be as in Definition \ref{dfn:(L)ISS-Lyapunov function}. 
Take an arbitrary $u \in U_{c}$ with $\|u\|_{U_{c}} \leq \rho_{u}$ and fix it. 
Consider
\[
I_t=\left\{x \in D:\|x\|_{X} \leq \rho_{x}, V(t,x) \leq \psi_{2} \circ \chi\left(\|u\|_{U_{c}}\right) \leq \rho_{x}\right\}.
\]

First, we show that $I_t$ is invariant, that is: $\forall x \in I_t \Rightarrow x(t)=\phi(t,t_0,x, u) \in I_t, t \geq t_0$.
If $I_t$ is not invariant, then, due to continuity of $\phi$ w.r.t. $t$, $\exists t_{*}>0$, such that $V\left(t_*,x\left(t_{*}\right)\right)=\psi_{2} \circ \chi\left(\|u\|_{U_{c}}\right)$, and therefore $\left\|x\left(t_{*}\right)\right\|_{X} \geq\chi\left(\|u\|_{U_{c}}\right)$.
The input to the system $\Sigma$ after time $t^{*}$ is $\tilde{u}$, defined by $\tilde{u}(\tau)=u\left(\tau+t^{*}\right), \tau \geq 0$.
According to the assumption of the theorem, $\|\tilde{u}\|_{U_{c}} \leq\|u\|_{U_{c}}$. 
Then from \eqref{eq:LISS implication form 1} it follows that $\dot{V}_{\tilde{u}}\left(t_*,x\left(t_{*}\right)\right)=-W\left(\left\|x\left(t_{*}\right)\right\|_{X}\right)<0$. 
Thus, the trajectory cannot escape the set $I_t$.

Second, we show that any trajectory starting outside $I_t$ must enter $I_t$ in finite time.
Take arbitrary $x_{0}:\left\|x_{0}\right\|_{X} \leq \rho_{x}$, and let $x(t)=\phi(t,t_0,x_0,u)$ be the trajectory starting at $x_0$. 
As long as $x_{0} \notin I_t$, we have $\exists\psi\in\mathcal{L}$ (depending on $W$) such that:
\[
\dot{V}(t,x(t)) 
\leq-\psi(\|x(t)\|_{X})\leq-\psi \circ \psi_{2}^{-1}(V(t,x(t))),\quad t\geq t_0,
\]
where $\psi\circ\psi_{2}^{-1}\in\mathcal{K}$.
It follows that $\exists \tilde{\beta} \in \mathcal{KL}: V(t,x(t)) \leq \tilde{\beta}\left(V\left(t_0,x_{0}\right), t-t_0\right)$, and consequently:
\begin{equation}\label{eq:ISS first part}
    \|x(t)\|_{X} \leq \beta\left(\left\|x_{0}\right\|_{X}, t-t_0\right), \quad \forall t: x(t) \notin I_t,
\end{equation}
where $\beta(r, t):=\psi_{1}^{-1} \circ \tilde{\beta}\left(\psi_{2}^{-1}(r), t\right), \forall r, t \geq 0$.
From the properties of $\mathcal{K} \mathcal{L}$ functions, it follows that $\exists t_{1}$:
\[
t_{1}:=\inf _{t \geq t_0}\left\{x(t)=\phi\left(t,t_0, x_{0}, u\right) \in I_t\right\}.
\]
From the invariance of the set $I_t$ we conclude that
\begin{equation}\label{eq:ISS second part}
    \|x(t)\|_{X} \leq \gamma\left(\|u\|_{U_{c}}\right), \quad t>t_{1},
\end{equation}
where $\gamma=\psi_{1}^{-1} \circ \psi_{2} \circ \chi \in \mathcal{K}$.
Our estimates hold for arbitrary control $u:\|u\|_{U_{c}} \leq \rho_{u} ;$ thus, combining \eqref{eq:ISS first part} and \eqref{eq:ISS second part}, we obtain the claim of the theorem. 
To prove the ISS of $\Sigma$ from existence of ISS-Lyapunov function, one can argue as above but with $\rho_{x}=\rho_{u}=\infty$.
\end{proof}

}


\bibliographystyle{IEEEtran}
\bibliography{References}

\begin{thebibliography}{10}
\providecommand{\url}[1]{#1}
\csname url@samestyle\endcsname
\providecommand{\newblock}{\relax}
\providecommand{\bibinfo}[2]{#2}
\providecommand{\BIBentrySTDinterwordspacing}{\spaceskip=0pt\relax}
\providecommand{\BIBentryALTinterwordstretchfactor}{4}
\providecommand{\BIBentryALTinterwordspacing}{\spaceskip=\fontdimen2\font plus
\BIBentryALTinterwordstretchfactor\fontdimen3\font minus
  \fontdimen4\font\relax}
\providecommand{\BIBforeignlanguage}[2]{{%
\expandafter\ifx\csname l@#1\endcsname\relax
\typeout{** WARNING: IEEEtran.bst: No hyphenation pattern has been}%
\typeout{** loaded for the language `#1'. Using the pattern for}%
\typeout{** the default language instead.}%
\else
\language=\csname l@#1\endcsname
\fi
#2}}
\providecommand{\BIBdecl}{\relax}
\BIBdecl

\bibitem{teodorovic2008swarm}
D.~Teodorovi{\'c}, ``Swarm intelligence systems for transportation engineering:
  Principles and applications,'' \emph{Transportation Research Part C: Emerging
  Technologies}, vol.~16, no.~6, pp. 651--667, 2008.

\bibitem{brambilla2013swarm}
M.~Brambilla, E.~Ferrante, M.~Birattari, and M.~Dorigo, ``Swarm robotics: a
  review from the swarm engineering perspective,'' \emph{Swarm Intelligence},
  vol.~7, no.~1, pp. 1--41, 2013.

\bibitem{crespi2008top}
V.~Crespi, A.~Galstyan, and K.~Lerman, ``Top-down vs bottom-up methodologies in
  multi-agent system design,'' \emph{Autonomous Robots}, vol.~24, no.~3, pp.
  303--313, 2008.

\bibitem{nouyan2008path}
S.~Nouyan, A.~Campo, and M.~Dorigo, ``Path formation in a robot swarm,''
  \emph{Swarm Intelligence}, vol.~2, no.~1, pp. 1--23, 2008.

\bibitem{hettiarachchi2009distributed}
S.~Hettiarachchi and W.~M. Spears, ``Distributed adaptive swarm for obstacle
  avoidance,'' \emph{International Journal of Intelligent Computing and
  Cybernetics}, vol.~2, no.~4, pp. 644--671, 2009.

\bibitem{cao2012overview}
Y.~Cao, W.~Yu, W.~Ren, and G.~Chen, ``An overview of recent progress in the
  study of distributed multi-agent coordination,'' \emph{IEEE Transactions on
  Industrial informatics}, vol.~9, no.~1, pp. 427--438, 2012.

\bibitem{accikmecse2012markov}
B.~A{\c{c}}ikme{\c{s}}e and D.~S. Bayard, ``A markov chain approach to
  probabilistic swarm guidance,'' in \emph{2012 American Control Conference
  (ACC)}.\hskip 1em plus 0.5em minus 0.4em\relax IEEE, 2012, pp. 6300--6307.

\bibitem{bandyopadhyay2017probabilistic}
S.~Bandyopadhyay, S.-J. Chung, and F.~Y. Hadaegh, ``Probabilistic and
  distributed control of a large-scale swarm of autonomous agents,'' \emph{IEEE
  Transactions on Robotics}, vol.~33, no.~5, pp. 1103--1123, 2017.

\bibitem{marden2009cooperative}
J.~R. Marden, G.~Arslan, and J.~S. Shamma, ``Cooperative control and potential
  games,'' \emph{IEEE Transactions on Systems, Man, and Cybernetics, Part B
  (Cybernetics)}, vol.~39, no.~6, pp. 1393--1407, 2009.

\bibitem{ferrari2006analysis}
G.~Ferrari-Trecate, A.~Buffa, and M.~Gati, ``Analysis of coordination in
  multi-agent systems through partial difference equations,'' \emph{IEEE
  Transactions on Automatic Control}, vol.~51, no.~6, pp. 1058--1063, 2006.

\bibitem{meurer2011finite}
T.~Meurer and M.~Krstic, ``Finite-time multi-agent deployment: A nonlinear pde
  motion planning approach,'' \emph{Automatica}, vol.~47, no.~11, pp.
  2534--2542, 2011.

\bibitem{qi2014multi}
J.~Qi, R.~Vazquez, and M.~Krstic, ``Multi-agent deployment in 3-d via pde
  control,'' \emph{IEEE Transactions on Automatic Control}, vol.~60, no.~4, pp.
  891--906, 2014.

\bibitem{pilloni2015consensus}
A.~Pilloni, A.~Pisano, Y.~Orlov, and E.~Usai, ``Consensus-based control for a
  network of diffusion pdes with boundary local interaction,'' \emph{IEEE
  Transactions on Automatic Control}, vol.~61, no.~9, pp. 2708--2713, 2015.

\bibitem{freudenthaler2020pde}
G.~Freudenthaler and T.~Meurer, ``Pde-based multi-agent formation control using
  flatness and backstepping: Analysis, design and robot experiments,''
  \emph{Automatica}, vol. 115, p. 108897, 2020.

\bibitem{milutinovi2006modeling}
D.~Milutinovi and P.~Lima, ``Modeling and optimal centralized control of a
  large-size robotic population,'' \emph{IEEE Transactions on Robotics},
  vol.~22, no.~6, pp. 1280--1285, 2006.

\bibitem{lasry2007mean}
J.-M. Lasry and P.-L. Lions, ``Mean field games,'' \emph{Japanese journal of
  mathematics}, vol.~2, no.~1, pp. 229--260, 2007.

\bibitem{hamann2008framework}
H.~Hamann and H.~W{\"o}rn, ``A framework of space--time continuous models for
  algorithm design in swarm robotics,'' \emph{Swarm Intelligence}, vol.~2, no.
  2-4, pp. 209--239, 2008.

\bibitem{foderaro2014distributed}
G.~Foderaro, S.~Ferrari, and T.~A. Wettergren, ``Distributed optimal control
  for multi-agent trajectory optimization,'' \emph{Automatica}, vol.~50, no.~1,
  pp. 149--154, 2014.

\bibitem{elamvazhuthi2015optimal}
K.~Elamvazhuthi and S.~Berman, ``Optimal control of stochastic coverage
  strategies for robotic swarms,'' in \emph{2015 IEEE International Conference
  on Robotics and Automation (ICRA)}.\hskip 1em plus 0.5em minus 0.4em\relax
  IEEE, 2015, pp. 1822--1829.

\bibitem{eren2017velocity}
U.~Eren and B.~A{\c{c}}{\i}kme{\c{s}}e, ``Velocity field generation for density
  control of swarms using heat equation and smoothing kernels,''
  \emph{IFAC-PapersOnLine}, vol.~50, no.~1, pp. 9405--9411, 2017.

\bibitem{krishnan2018distributed}
V.~Krishnan and S.~Mart{\'\i}nez, ``Distributed control for spatial
  self-organization of multi-agent swarms,'' \emph{SIAM Journal on Control and
  Optimization}, vol.~56, no.~5, pp. 3642--3667, 2018.

\bibitem{elamvazhuthi2018bilinear}
K.~Elamvazhuthi, H.~Kuiper, M.~Kawski, and S.~Berman, ``Bilinear
  controllability of a class of advection--diffusion--reaction systems,''
  \emph{IEEE Transactions on Automatic Control}, vol.~64, no.~6, pp.
  2282--2297, 2018.

\bibitem{lieberman1996second}
G.~M. Lieberman, \emph{Second order parabolic differential equations}.\hskip
  1em plus 0.5em minus 0.4em\relax World scientific, 1996.

\bibitem{sontag1995characterizations}
E.~D. Sontag and Y.~Wang, ``On characterizations of the input-to-state
  stability property,'' \emph{Systems \& Control Letters}, vol.~24, no.~5, pp.
  351--359, 1995.

\bibitem{dashkovskiy2013input}
S.~Dashkovskiy and A.~Mironchenko, ``Input-to-state stability of
  infinite-dimensional control systems,'' \emph{Mathematics of Control,
  Signals, and Systems}, vol.~25, no.~1, pp. 1--35, 2013.

\bibitem{evans1998partial}
L.~C. Evans, \emph{{Partial differential equations}}, ser. Graduate Studies in
  Mathematics.\hskip 1em plus 0.5em minus 0.4em\relax Providence, RI: American
  Mathematical Society, 1998.

\bibitem{silverman1986density}
B.~W. Silverman, \emph{Density Estimation for Statistics and Data
  Analysis}.\hskip 1em plus 0.5em minus 0.4em\relax CRC Press, 1986, vol.~26.

\bibitem{curtain1995introduction}
R.~F. Curtain and H.~Zwart, \emph{An Introduction to Infinite-Dimensional
  Linear Systems Theory}.\hskip 1em plus 0.5em minus 0.4em\relax Springer
  Science \& Business Media, 1995, vol.~21.

\end{thebibliography}

\end{document}